\documentclass[a4paper,12pt]{article}
\usepackage{bm}
\usepackage{amsmath}
\usepackage{amssymb}
\usepackage[dvips]{graphicx}
\usepackage{pdflscape}
\usepackage{multirow}
\usepackage[sort]{natbib}
\usepackage{soul}
\usepackage{amsmath}
\usepackage{graphicx}
\usepackage{enumerate}
\usepackage{natbib}
\usepackage{url} 

\usepackage[utf8]{inputenc}
\usepackage{mathtools}
\usepackage{amsbsy}
\usepackage{amssymb}
\usepackage{color}
\usepackage{ulem}
\usepackage{placeins}
\usepackage{booktabs}
\usepackage{multirow}
\usepackage{hyperref}
 \usepackage{booktabs,subcaption,amsfonts,dcolumn}
\usepackage{xcolor,colortbl}
\usepackage{subcaption}
\usepackage{verbatim}
\usepackage{chngcntr}
\usepackage{apptools}
\usepackage{comment}
\usepackage{movie15}
\usepackage{epstopdf}
\AppendGraphicsExtensions{.gif}

\AtAppendix{\counterwithin{lem}{section}}
\AtAppendix{\counterwithin{cor}{section}}

\usepackage{amsthm}

\newtheorem{prop}{Proposition}

\newtheorem{rem}{Remark}

\definecolor{Gray}{gray}{0.85}
\definecolor{LightCyan}{rgb}{0.88,1,1}
 
\newcolumntype{b}{>{\columncolor{Gray}}c}
\newcolumntype{a}{>{\columncolor{red}}c}
\newcolumntype{g}{>{\columncolor{green}}c}

\title{Movement of insurgent gangs: A Bayesian kernel density model for incomplete temporal data }
\author{ Karthik Sriram\footnote{Associate Professor, Indian Institute of Management  Ahmedabad, Gujarat, India (karthiks@iima.ac.in)}~~~Dhruv Gupta  ~~~Rajiv Parikh  }
\begin{document}
\date{\today} \maketitle
\begin{abstract}
We develop a Bayesian modeling framework to address a pressing real-life problem faced by the police in tackling insurgent gangs. Unlike criminals associated  with common crimes such as robbery, theft or street crime,  insurgent gangs  are trained in sophisticated arms and strategise against the government to weaken its resolve. They are constantly on the move, operating over large areas causing damage to national properties and terrorizing ordinary citizens. Different from the more commonly addressed problem of modeling crime-events, our context requires that an approach be formulated to model the movement of insurgent gangs, which is more valuable to the police forces  in preempting their activities and nabbing them. This paper evolved as a collaborative work with the Indian police to help augment their tactics with a systematic method, by integrating past data on observed gang-locations with the expert knowledge of the police officers. A methodological challenge in modeling the movement of insurgent gangs is that the data on their locations is incomplete, since they are observable only at some irregularly separated time-points. Based on a weighted kernel density formulation for temporal data, we analytically derive the closed form of the likelihood, conditional on incomplete past observed data. Building on the current tactics used by the police, we device an approach for constructing an expert-prior on gang-locations,  along with a sequential Bayesian procedure for estimation and prediction. We also propose a new metric for predictive assessment that complements another known metric used in similar problems.
\end{abstract}

\section{Introduction}
Managing insurgencies is a herculean task for local security forces in many countries. For example, India, Democratic Republic of Congo, Somalia and Yemen are tackling insurgent gangs for several years.  Unlike criminals associated  with common crimes such as robbery, theft or street crime,  insurgent gangs  deliberately strategise against the government to weaken its resolve. Therefore, the gangs train their cadre well in sophisticated firearms and often engage in gun-fights with the security forces to establish dominance. The gangs usually focus their operations in some particular regions of the country targeting different locations over time, causing damage to national properties such as roads, bridges, factories, police facilities and terrorising ordinary citizens. When not attacking, they are constantly on the move to avoid being caught or neutralized by the security forces. Therefore, it is  essential that the security forces anticipate the movements of the insurgent gangs in order to  proactively tackle them.  Typically, crime-modeling has focused on occurrence of crime-events, often overlooking the particular organized gang that committed the crime. Modeling crime-events may help the police prevent occurrence of crime in some hotspots, but may not necessarily lead to nabbing the criminals, which is their ultimate goal. In our context of insurgent gangs, trying to predict crime-events would be inadequate as well as difficult, given that such events are sparsely distributed over time.  Often challenged with limited resources, the security forces require novel methods of intervention to be effective. Hence, modeling the movement of insurgent gangs will be valuable to the police forces  in preempting criminal activities and nabbing them.

This paper evolved as a collaborative work with the Indian police, as  one of the authors is a senior police officer dealing with insurgent gangs. Currently,  the operations of the Indian police are largely driven by their experience. It is important to augment their tactics with a systematic method that can integrate past data on observed locations with the expert knowledge of the police officers, to improve the efficiency of using intelligence inputs and quality of patrols. Without a systematic method, most encounters with the gangs are only by chance.  Hence, our objective is to  develop a statistical modeling framework that can incorporate historical data on gang locations as well as the expert knowledge of police officers based on their experience and informant network, for predicting the movement of some insurgent  gangs in India. Since our modeling methodology builds on the tactics already used by the police forces during their counter-insurgency operations, it is amenable to easy adoption. Further, with the passage of time, the police forces will be able to gradually improve the  quality of information gathered, as well as the completeness of the data required for the model.

  As mentioned above, a different but related problem considered more often in the literature is that of modeling locations of crime-events. The difference is that gang-locations are not always accompanied by crime events. So,  the time-points other than those at which we have observed the gang-locations, are essentially time-points with missing data on locations.  The approaches to model crime-event locations include point-process models, e.g. log-Gaussian Cox process  (\citealt{shirota_gelfand2017}, \citealt{flaxmanetal2019}), self-exciting processes (e.g. \citealt{mohleretal2011}, \citealt{zhuang_mateu2019}),  auto-regressive mixture models (e.g. \citealt{taddy2010}).  Also, kernel density estimation (see  \citealt{silverman1986}) has been used extensively to develop  methods for generating hotspot maps to identify likely event locations (e.g. \citealt{brunsdon2007}, \citealt{Huetal2018}, \citealt{yuanetal2019}).  Relatively fewer works consider data on specific offenders (e.g. \citealt{porter_reich2012} ,  \citealt{liwang2018}) but the focus is to model crime events rather than offender movements. Some works do consider offender mobility (e.g. \citealt{brantingham_tita2008}, \citealt{rosesetal2020}), though not gang-movements, with the objective of generating hotspot maps for crime-events resulting as an aggregate of simulated behaviour of many offenders. Notably, a similar problem that has received attention is of animal-movements, through model formulations based on integrated Brownian motion (e.g. \citealt{hooten_johnson2017}), correlated random walk formulations (e.g. \citealt{wang2019}, \citealt{buttsetal2022}), and non-homogenous spatial Poisson processes (e.g. \citealt{hanksetal2015}). However, random walk formulations typically assume independent increments, which is not realistic in the context of gang movements, where the gangs show long term path dependence. Stochastic formulations such as non-homogeneous Poisson process, where the rate function  can potentially depend on all past locations, are more amenable to data that are frequently recorded at regular time-points. They become mathematically intractable if data is incomplete, without being observed at all the desired time-points. Unarguably, regular and frequent recording of locations may be possible in the context of animal-movements, for example by using tracking collars, but certainly not possible with insurgent gangs. To cater to these issues, we formulate a nonparametric model using a kernel density approach that adequately accounts for incompleteness of the temporal data  observed at irregularly separated time points.
   
Further, our context requires a systematic approach to incorporate expert knowledge of police officers involved in anti-insurgent operations because locations of  insurgent gangs are observable only at irregularly separated time-points and spread over large areas. To illustrate this,  Figure \ref{F:fig2}  (i) shows the first 20 instances (in our data) of the observed locations for a particular gang (which we refer to as gang A) over a span of 500 days. The locations of the gang are  usually in the forests and away from the police camps. Figure \ref{F:fig2}  (ii) shows the gaps in days between the successively observed instances of gang A. It can be seen that gaps between observed instances can easily be more than 20 days apart and sometimes more than 50 days. The area of movements of gang A in our data happens to be over 5500 sq km but the gang could continue expanding its operations covering a larger area, based on police action or in search of resources to sustain. Therefore, a logical methodology that can incorporate any special features to help predict the location of insurgent gangs becomes relevant. In our approach, we complement the observed data on gang-locations with expert knowledge of police officers, which is a combination of their knowledge of the gang's recent past locations, some insights about the gang's preferences (e.g.  high density forest and non-proximity to police camps), and  intelligence inputs through informant networks. 
 \begin{figure}[ht]
\center
\caption{\label{F:fig2} For the first 20 observed instances of insurgent gang A in our data, part (i)  shows the successive observed locations, and part (ii) shows the gaps in days between successively observed instances. The gang operates in the Jharkhand state of India, and its span of area in these 20 instances is approximately 5500 sq km.  The grey shaded parts indicate forest areas and triangle-symbols indicate police camps.  Due to sensitivity of the information, the longitudes and latitude values are omitted on the axes. }
\begin{subfigure}{.48\textwidth}
\caption{successive observed locations }
\includegraphics[width=6cm, height=6cm,angle=0]{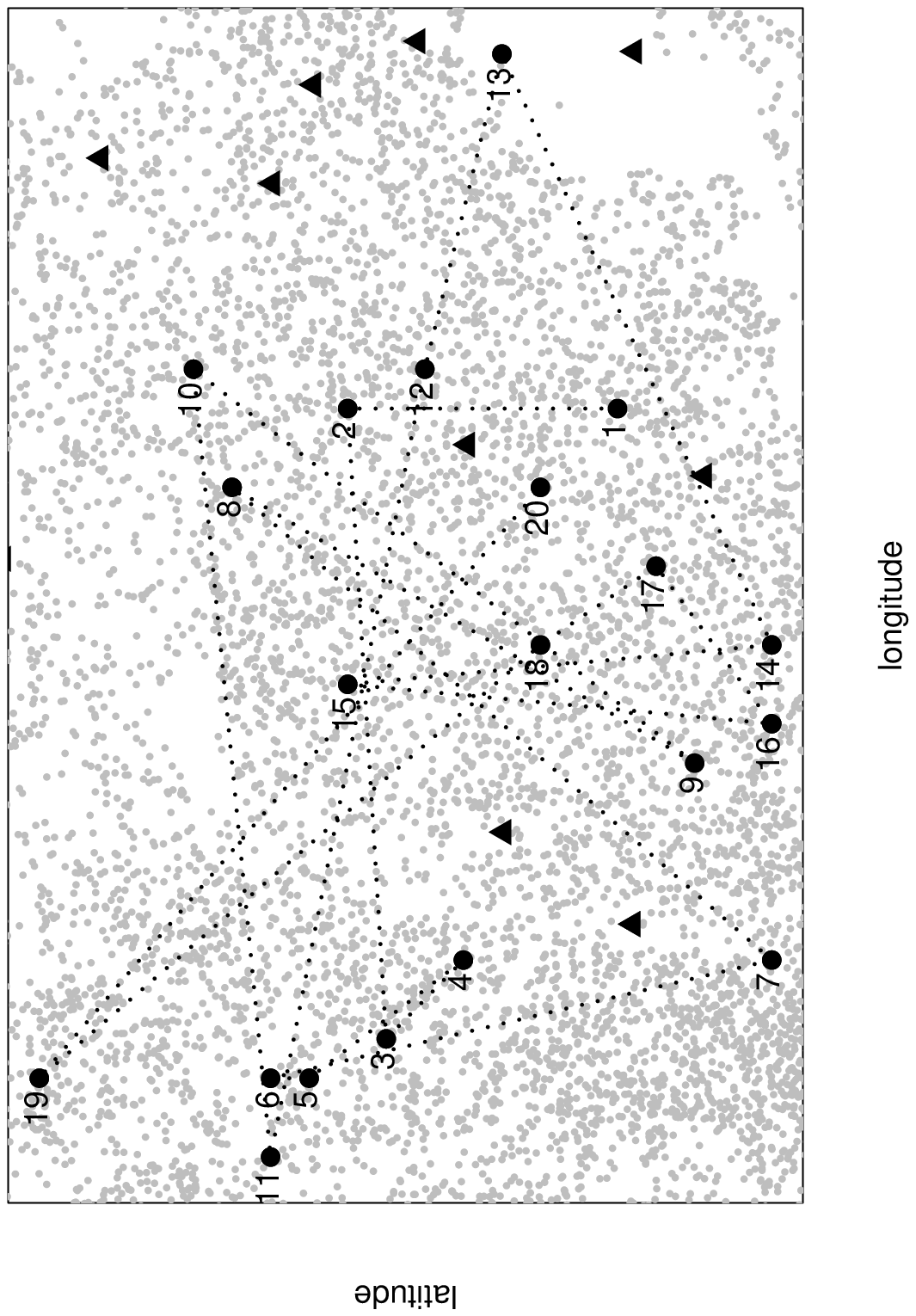}
\end{subfigure}
\begin{subfigure}{.48\textwidth}
\caption{days between observed instances }
\includegraphics[width=6cm, height=6cm,angle=270]{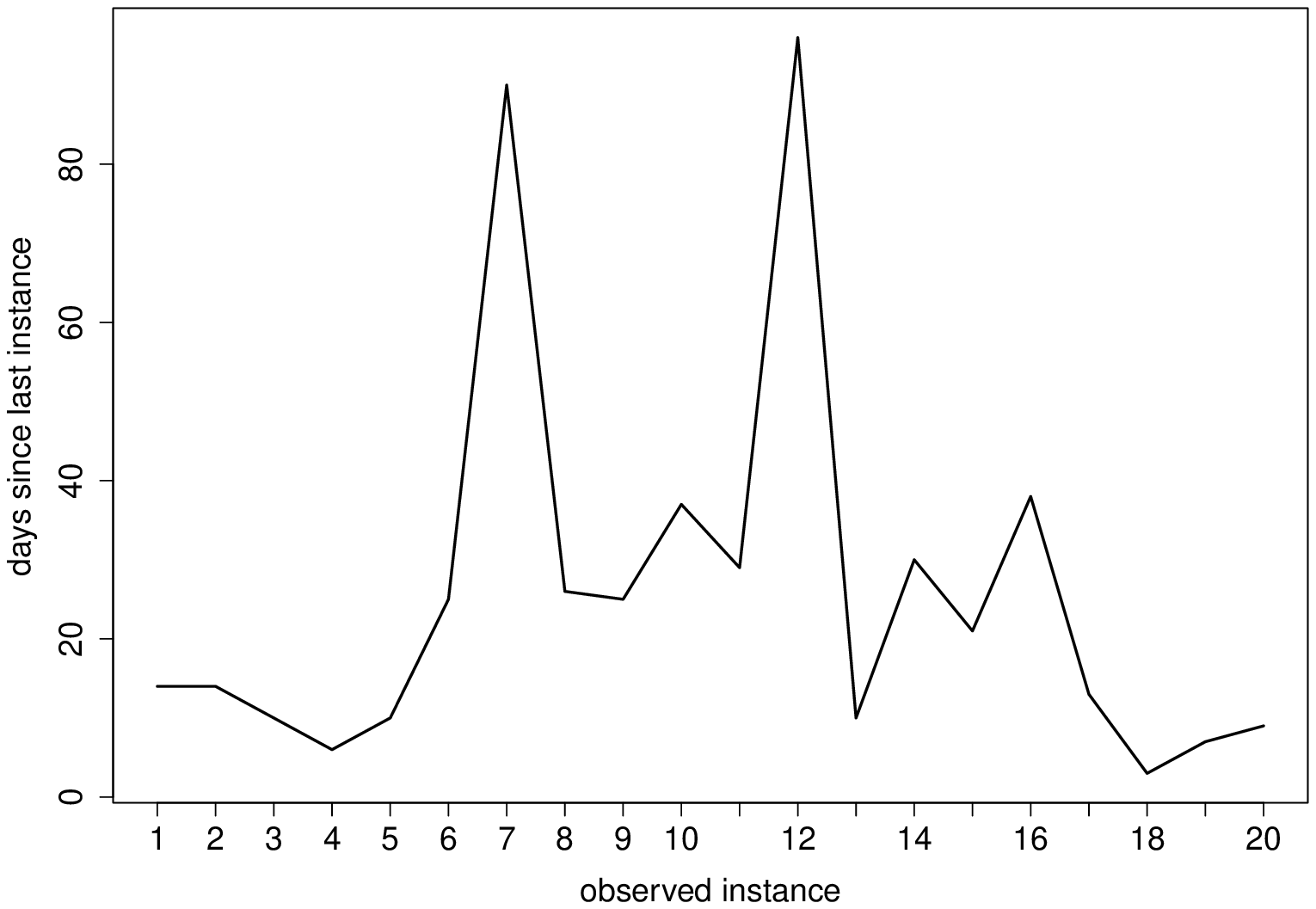}
\end{subfigure}
\end{figure} 

 ~\\
In our paper, we develop a Bayesian modeling framework to predict the locations of insurgent gangs, based on historical data and expert-knowledge of the police officers. In particular, our contributions include the following:
 ~\\
 (i) {\bf A weighted kernel density model for temporal data (on gang-locations) observed at discrete time-points (i.e. days).}  The idea is to  update the model every day and be able to predict the gang-location for the next day, conditioned on the observed locations till that day.  However, a methodological challenge is that data on locations is incomplete, i.e. they are observed only on some irregularly separated days.  As a first step, in Proposition \ref{thm:result1},   we analytically derive the closed form of the model likelihood when conditioned on the incomplete past data.

 ~\\
 (ii) {\bf  A systematic approach to dynamically construct and incorporate a prior based on expert knowledge of the police officers.} A methodological challenge is that on any given day, the expert-prior is stated on the ``observables", i.e.  as  a spatial probability density for possible gang-locations, and not directly as a probability distribution on the model parameters.  Further, on any given day, not only do we have a prior on the model parameters obtained through Bayesian updates on past days, but also the expert-prior on gang-locations that is specifically constructed for the current day.   We device a dynamic Bayesian approach for estimation and predictions while addressing these challenges. 
 
 ~\\
 (iii) {\bf A new metric for predictive assessments.}  A natural metric that can be considered in our problem is the ``Required Area to Monitor" (RAM), an idea used by \citealt{porter_reich2012}, which is the area to be monitored in decreasing order of predicted model likelihood till the actual gang-location. To complement RAM, we propose the idea of area under the `Proximity Curve', which plots percentage of area to monitor versus the minimum distance to the actual gang-location, of locations within that area.
 
 ~\\
In what follows, Section \ref{S:background} describes the data preparation.   Section \ref{S:model_expert} gives the formulation of our main model and the construction of the expert-prior. Section \ref{S:est_pred} gives our sequential Bayesian approach to estimation and prediction along with a description of some important metrics for predictive assessment. Section \ref{S:results}   discusses the results of fitting the model to data on different insurgent gangs in India. We  conclude with a discussion in Section \ref{S:conclusion}. Details on some modeling procedures,  proof of the main mathematical result and some simulation results are deferred to the Appendix, in order to not obstruct the flow of the paper.
\section{Data preparation}
\label{S:background}
For our work, we use data on some major insurgent gangs that operate in the state of Jharkhand in India. They are referred to as ``Naxalite gangs" in India. Naxalite insurgency is a serious armed conflict, involving radical communists organised as multiple insurgent gangs. It began in the year 1967, and derives its name from ``Naxalbari", a village where farmers had protested for land rights. However,  over time, it gradually evolved into an armed conflict where several criminal gangs emerged, and continue to engage in violence leading to several fatalities of police and civilians, as well as destruction of national property. \cite{dhruv_karthik2018} give a more detailed background of the naxalite problem in India.   The Indian government has deployed special federal security forces called the Central Reserve Police Force (CRPF) in aid of state police forces to help tackle this insurgency.   The security forces regularly patrol areas affected by naxalite activity to help reduce the casualties and destruction to national assets. Our modeling approach will use data on historical locations of different gangs and also information relating to the expert knowledge of the police,  towards constructing a Bayesian prior base on variables such as forest density, proximity to CRPF camps and intelligence inputs through the informant network. The variables we use for the construction of the prior are only an indicative list but can serve as a good starting point towards designing better data gathering systems in the future.  We will give the details on the approach for the model formulation and constructing the expert prior in  Section \ref{S:model_expert}. Below, we elaborate on the related data preparation exercise.

~\\
 A. {\bf Past observed locations of the gang and corresponding dates:} For the different insurgent gangs we obtain the past observed geographical coordinates (longitude, latitute) along with the dates, based on  police records for the period Jan 2019 to June 2022. Typically,  location of a gang is known when there is a chance encounter with the police, or when the gang causes some notable disturbance such as killings, arson, bombing or extortions. During investigations, subsequent to the incident, the police determine and record the particular gang that was involved. 
 
~\\
B. {\bf Locations of CRPF camps:}  Generally, insurgent gangs  prefer to move in areas that are not close to the CRPF camps. Several CRPF camps are located across the state of Jharkhand, and  some of them are illustrated using triangle-symbols in Figure \ref{F:fig2}(i).  We will use the CRPF camp locations obtained from the police department, to identify locations in close proximity to the CRPF camps, while constructing the expert prior.

~\\
C. {\bf Deriving forest density at locations, using satellite images:}  An input to the expert prior is the forest density at any given location.  For practical purposes, we consider a fine grid of points covering the Jharkhand state and classify each point in the grid as forest or non-forest. Then, we can approximately compute the forest density at any location as the percentage of grid points within a suitable radius that are identified as forest. However, whether any given location is within the forest is itself not readily known, and requires a separate ``landuse modeling" exercise, which we carry out as part of this data preparation exercise. For the landuse modeling, we first obtain `landuse satellite' (landsat) images covering the state of Jharkhand, India from the U.S. Geological Survey ( \href{https://earthexplorer.usgs.gov}{USGS}). We then develop a classification model for each landsat image to classify different locations within the image as either forest or non-forest, given the  latitude and longitude using the `superclass' package in R (\citealt{cran}).  The detailed approach to landuse modeling is described in  Appendix-Section \ref{S:landuse_model}. 

~\\
D. {\bf Informant Intelligence}:  Through their informer networks, the police forces are on a constant lookout for `tell-tale' signs which indicate possible presence of naxalite gangs in an area. Such information is usually recorded in a police diary in Hindi or English or vernacular languages. Table \ref{T:intel} gives a couple of examples of how such intelligence information is recorded by the police (after translating to English and masking the identities of the gangs and locations due to the sensitive nature of the information). Whenever such information was available, we manually looked up the suggested locations (as close as possible) and also noted down the date when such intelligence was made available. 
\begin{table}[htbp]
  \centering
  \caption{\label{T:intel} Illustration of informant intelligence information recorded in the police diaries after masking the names of gangs and locations }
    \begin{tabular}{|l|p{10cm}|}
    \hline
     example  & intelligence description \\
     \hline
 example 1 & An armed group (12-15) of *** group led by *** has been noticed camping in the forest area located in between villages ***  and ***. The location of the group is south of the  school ****. \\
\hline
example 2 & A group (15-20) of *** group led by *** is camping in the forest between *** village and *** village, 1.5 km north of *** hills,  bordering police station *** near the streamlet *** \\
\hline
    \end{tabular}%
\end{table}%
\FloatBarrier

\section{Formulation of the model and the expert-prior}
\label{S:model_expert}
In this section, we  formulate a weighted kernel density model for the gang-location on a given day, conditional on incomplete past data on gang-locations, i.e. when the past gang-locations are observed only on irregularly separated days. We also propose an approach to construct an expert-prior map for the possible location of a gang on any given day, based on expert knowledge of the police.  Estimation and prediction with these formulations will be discussed in the Section \ref{S:est_pred}.  We denote the locations on days $1$ to $n$ by $s_{1:n}$, the likelihood of location on day $n$ conditional on all the past locations by $f_\Theta\left(s_{n}\vert s_{1:(n-1)}\right)$,  and  the expert-prior map for the gang-location on day $n$, by $E_n(s_n)$. Understandably, the locations, likelihood and expert-prior map are gang-specific.  So, for a gang $g$, this could be explicitly shown, for example, by writing  $s^{(g)}_{1:n}$,  $f^{(g)}_{\Theta}$ and $E^{(n,g)}$, respectively. However, for notational simplicity, we suppress their dependence on the gang $g$, with the understanding that on any given day, we will carry out the sequential Bayesian update for the model parameters  by considering the observed location of one gang at a time, along with the likelihood, past data and expert-prior for that particular gang.

\subsection{Model}
\label{S:model}
We consider discrete time points in ``days" and denote the locations of a given gang on successive days, numbered $\{1, 2, \ldots, n, \ldots\}$ by $\{s_1, s_2, \ldots, s_n, \ldots \}$, where each $s_j=(s_{jx}, s_{jy})$ gives the (longitude, latitude) coordinates of the gang-location on the $j^{th}$ day. The gang locations in our context are to be tracked in the state of Jharkhand covering an area of about 80,000 sq km and the security forces may only get to know their locations in gaps of several days. Given the large area involved and the uncertainty involved in knowing any gang's location in real-time, it is practically reasonable to work with one representative location of the gang's movement on any given day.  So, the location on any given day in our data is to be interpreted as an approximate location representing the gang's position on that day.  A reasonable model formulation should consider the fact that the future movement locations are likely to depend on many of the previously visited locations. Accordingly, we may denote the probability density model for the gang-location on the $(n+1)^{th}$ day conditional on the locations of the past $n$ days by 
\begin{equation}
f_\Theta\left(s_{n+1}\vert s_{1:n}\right), ~n \geq 1,\label{eq:lik1}
\end{equation}
where $\Theta$ are the model parameters and $s_{1:n}=(s_1, s_2, \ldots, s_n)$. However, it is important to note that the locations of the gang may not be observed on several days. We denote the index of the days when the location is not observed by $u_{1:L}=\{u_1, u_2,\ldots,  u_L\}$, where $1< u_1<u_2<\cdots u_L \leq n$. Equivalently, the gang-locations are only observed on some irregularly separated days, viz.  $(1:n)\backslash u_{1:L}$. So, we can denote the observed data as of the $n^{th}$ day by 
\begin{eqnarray}
\mathcal{D}_n&&= s_{(1:n)\backslash u_{(1:L)}}\nonumber\\
&&= \left\{s_1,\ldots, s_{u_1-1},\square , s_{u_1+1},\ldots, s_{u_2-1}, \square , s_{u_2+1}\ldots, s_{u_L-1}, \square,s_{u_L+1}, \ldots, s_n \right\}, \label{eq:D_n}
\end{eqnarray}
where the squares ($\square$) indicate the missing time points. We note that equation (\ref{eq:D_n}) is a general notation and does not preclude the possibility that data may not be observed for a continuous run of several days.  It is clear that the value of $L$ has to be less than $n$ but we suppress this dependence to avoid cumbersome expressions.  In order to model  data that are observed at irregular time points, we need to derive the  likelihood for the location on  $(n+1)^{th}$ day, conditional only on the observed data, starting from the model specification (\ref{eq:lik1}), as follows: 
\begin{eqnarray}
&&f_\Theta(s_{n+1}\vert \mathcal{D}_n)= f_\Theta\left(s_{n+1}\vert s_{(1:n)\backslash u_{(1:L)}}\right) \nonumber\\
 &&=\int\ldots\int f_\Theta\left(s_{n+1}\vert s_{(1:n)}\right)
\cdot f_\Theta\left(s_{u_1}\vert s_{(1:(u_1-1))}\right)\cdots f_\Theta\left(s_{u_L}\vert s_{(1:(u_L-1))\backslash u_{1:(L-1)}}\right)\nonumber\\
&~&\hspace{11cm}{ds_{u_1}\ldots ds_{u_L}.}\label{eq:lik2}
\end{eqnarray}
In general, deriving the likelihood (\ref{eq:lik2}) starting from (\ref{eq:lik1}) is non-trivial. Of course, it could be obtained under some simplifying  assumptions such as Markovian or multivariate Gaussian, but those could be unreasonably restrictive in practice. A Markovian assumption restricts the likelihood of future locations to depend only on the most recently observed location, and a Gaussian model imposes  a unimodal likelihood structure. Instead, we propose a non-parametric weighted kernel density approach, and analytically derive the solution to equation (\ref{eq:lik2}). The non-parametric nature of the kernel density provides the flexibility to capture dependence on all historical data, along with the possibly multi-modal spatial nature of the probability density that we are trying to predict, while being amenable to a closed-form expression for the conditional likelihood (\ref{eq:lik2}).

Our model formulation for the probability density for location $s_{n+1}$ on day $(n+1)$, conditional on locations $s_1, s_2, \ldots, s_n$ on all past $n$ days (i.e. if they were all observed), is given by the following weighted kernel density formulation:
\begin{eqnarray}
f_\Theta\left(s_{n+1} \vert s_{1:n}\right)=\sum_{i=1}^n w_{\theta}(i, n+1) \kappa_2(s-s_i; \underline{h}), n\geq 1, \label{eq:fullmod}
\end{eqnarray}
where  $\Theta=(\theta, \underline{h})$ is the parameter vector taking values in a suitable parameter set $\mathcal{S}_\Theta$, $\kappa_2(z, \underline{h})$ is the bivariate standard Gaussian kernel:
 \begin{eqnarray}
  \kappa_2(z, \underline{h})= \frac{1}{\sqrt{2\pi}h_1} e^{-\frac{z_1^2}{2 h_1^2}}\cdot  \frac{1}{\sqrt{2\pi}h_2} e^{-\frac{z_2^2}{2 h_2^2}}, \mbox{ for } z=(z_1,z_2),\underline{h}=(h_1, h_2), \label{eq:bivnormal}
  \end{eqnarray}
and $w_\theta(i, n+1)$ is a time-weighting function parametrised by $\theta$, such that 
\begin{eqnarray}
\mbox{ for any } n,~  \sum_{j=1}^n w_\theta(j, n+1)=1. \label{eq:weight_cond} 
\end{eqnarray}

Bivariate kernel density formulations have been used for spatial modeling  of crime (e.g.  \citealt{brunsdon2007}, \citealt{yuanetal2019}).  Our formulation (\ref{eq:fullmod}) is a  time-weighted bivariate kernel density.  Weighted terms in kernel density formulations have been used in other problems, e.g. \citet{harvey_orysh2012}  for time-series data,   \cite{Gao_Zhong2010}  use random uniformly distributed weights,  \cite{wang_wang2007} consider deterministic weights  based on a covariates, and  \citet{takde2022} use temporally block-weighted kernel density for  streaming data. \citet{porter_reich2012}  use a time-weighted kernel density formulation and we will note its connection to our work in Remark \ref{rem1}. 

Equation (\ref{eq:fullmod}) leads to a full stochastic specification because if we knew $f_\Theta(s_1)$ then the joint pdf of $(s_1,s_2,\ldots, s_{n+1})$ for every $n$  can then be written as
\begin{equation}
f_\Theta(s_1)\prod_{k=2}^{n}f_\Theta\left(s_{k+1} \vert s_{1:(k-1)}\right). \label{eq:jointpdf}
\end{equation}
 In the interest of parsimony, we simplify the bandwidth parameters as
 \begin{eqnarray}
 \underline{h}=(h_1,h_2)=h \cdot (\delta_{lon} , \delta_{lat}), ~h>0, \label{eq:h1h2}
 \end{eqnarray}
  where the values of $(\delta_{lon}, \delta_{lat})$ are chosen to correspond to a 1 km distance along the respective axes, and $h$ is a free scalar-multiple that will be estimated. The time-weighting function can be specified in terms of a parametric univariate pdf $g_\theta(t)$, $t\in [0, \infty)$, as
      \begin{eqnarray}
 w_\theta(i, n+1)=\frac{g_{\theta}(n+1-i)}{\sum_{j=1}^{n}g_{\theta}(n+1-j)}. \label{eq:weight}
 \end{eqnarray}
 The nature of $g_\theta$ chosen would depend on the specific application. For example, if $g_\theta(t)$ is a unimodal pdf on $t\in [0, \infty)$ with mode at $t=0$ and $\theta$ being a parameter relating to its spread, then it would result in a time-decaying weight function. The model requires the estimation of the following parameters
  \begin{eqnarray}
  \Theta= (\theta, h), \theta\geq 0,~h>0. \label{eq:theta1}
  \end{eqnarray} 
For our work, we take
 \begin{eqnarray}
  g_{\theta}(t)= e^{-t/\theta} \implies  w_\theta(i, n+1)=\frac{e^{-(n-i)/\theta} \left(1-e^{1/\theta}\right)}{1-e^{-n/ \theta}}, ~ \mbox{ for } i\leq n. \label{eq:exponential}
  \end{eqnarray}
In (\ref{eq:exponential}), the parameter $\theta$ controls the rate of decay of the weights as the time difference increases. Very large values of $\theta$ will correspond to approximately uniform weighting and small values will imply that only the very recent past events will get weighted.    

 As in equation (\ref{eq:D_n}), if the data until day $n$ is only observed on days $(1:n)\backslash u_{1:L}$, then we need to solve the problem of obtaining the likelihood
 \begin{eqnarray}
 f_\Theta\left( s_{n+1}\vert \mathcal{D}_n\right)= f_{(\theta,h)}\left(s_{n+1}=s\vert s_{(1:n)\backslash u_{(1:L)}}\right) \label{eq:proposed_model},
 \end{eqnarray}
  which can be obtained using (\ref{eq:lik2}) starting from (\ref{eq:fullmod}). The following proposition gives the solution to this problem.
 \begin{prop}
 \label{thm:result1}
Consider the weight function $w_\theta(\cdot, \cdot)$ as in (\ref{eq:weight}), $\underline{h}$ as in equation (\ref{eq:h1h2}). Among data points $1:n$, say $u_1, \ldots, u_L$ are missing and let
 \begin{eqnarray}
&& \mathcal{S}_1=\{\},~\mathcal{S}_2=\{u_1\}, \ldots , \mathcal{S}_L=\{u_1,\ldots, u_{L-1}\}, ~\mathcal{S}_{L+1}=\{u_1,\ldots, u_{L}\}.\label{eq:S}
 \end{eqnarray}
 Then, the likelihood (\ref{eq:lik2}) is given by  
 \begin{eqnarray}
 &&~\begin{split}
&  f_{\Theta}\left(s_{n+1}=s\vert \mathcal{D}_n\right)=f_{\Theta}\left(s_{n+1}=s\vert s_{(1:n)\backslash \mathcal{S}_{L+1}}\right)\nonumber \\
&= \sum_{i\in(1:n)\backslash\mathcal{S}_{L+1}} w_\theta(i,n+1) \kappa_2(s-s_i,\underline{h}) \\
&+\sum_{m=1}^L \left( \sum_{1\leq i_1 < i_2\cdots<i_m\leq L} w_\theta(u_{i_1},u_{i_2}) \cdot w_\theta(u_{i_2},u_{i_3}) \cdots \right. \\ 
& ~~~~\cdots w_\theta(u_{i_{m-1}},u_{i_m}) \cdot w_\theta(u_{i_m},n+1) \cdot \left.  \sum_{j\in(1:u_{i_1}-1)\backslash\mathcal{S}_{i_1}} w_\theta(j,u_{i_1}) \kappa_2(s-s_j,\sqrt{m+1} \cdot \underline{h})\right).
 \end{split}\nonumber\\
 &&~ \label{eq:result1}
 \end{eqnarray}
 \end{prop}
 Proof of Proposition \ref{thm:result1} uses mathematical induction and is given in  Appendix-Section \ref{S:app_proof}.  Towards obtaining a computationally more convenient form of  (\ref{eq:result1}), let $I_{\{condition\}}$ denote the indicator function which takes the value $1$ or $0$ depending on whether the $\{condition\}$ is satisfied or not.  Then, we define the following matrices:
\begin{eqnarray}
  A_\theta&:=& \left(\left(A_\theta(p,q)  \right)\right)_{p=1:n, q=1:L} , \mbox{ where }A_\theta(p, q) =I_{\left\{ p\in(1:u_{q}-1)\backslash\mathcal{S}_{q}\right\}} \cdot w_\theta(p, u_{q}),\\
  W_\theta&:=& \left(\left(W_\theta(p,q)\right)\right)_{p=1:L, q=1:L} ,\mbox{ where }W_\theta(p,q) =I_{\{ p<q\}} \cdot w_\theta(u_{i_p},u_{i_{q}}),\\
B_\theta &:=& \left(\left(B_\theta(p)\right)\right)_{p=1:L} ,\mbox{ where }  B_\theta(p) =w_\theta(u_{i_p}, n+1) ,\\
C^{(m)}_\theta&:=& A_\theta \times \left(W_\theta\right)^{m-1} \times B_\theta= \left( \left( C^{(m)}_\theta(j)\right)\right), j=1:n.
  \end{eqnarray}
Then, note that by rearranging summations and using indicator functions,  the second term of the right hand side of equation (\ref{eq:result1}) can be written as
  \begin{eqnarray}
  &&\begin{split}
& \sum_{m=1}^L \sum_{j=1}^n \left\{ \sum_{i_1=1}^L\sum_{i_2=1}^L\ldots \sum_{i_m=1}^L  A_\theta(j, i_1)\cdot W_\theta(i_1, i_{2})\cdots \right. \\
& \left. ~~~~~~~~~~~~~~~~\cdots W_\theta(i_{m-1},i_{m})\cdot B_\theta(i_m)\right\}\kappa_2(s-s_j,\sqrt{m+1} \cdot \underline{h}).
\end{split}~\label{eq:step1}
  \end{eqnarray}
As  a consequence, we obtain the following proposition.
 \begin{prop}
 \label{thm:result2}
 Consider the weight function $w_\theta(\cdot, \cdot)$ as in (\ref{eq:weight}), $\underline{h}$ as in equation (\ref{eq:h1h2}), and  equations (\ref{eq:S}) to  (\ref{eq:step1}). Then, the likelihood (\ref{eq:lik2}) can be written as
 \begin{eqnarray}
 &&  f_{\Theta}\left(s_{n+1}=s\vert \mathcal{D}_n\right)=f_{\Theta}\left(s_{n+1}=s\vert s_{(1:n)\backslash \mathcal{S}_{L+1}}\right)\nonumber \\
 && =\sum_{i\in(1:n)\backslash\mathcal{S}_{L+1}} w_\theta(i,n+1) \kappa_2(s-s_i,\underline{h})  +\sum_{m=1}^L \sum_{j=1}^n  C^{(m)}_\theta(j)\kappa_2(s-s_j,\sqrt{m+1} \cdot \underline{h}).\nonumber\\
 &&~\label{eq:result2}
  \end{eqnarray}
 \end{prop}
Interestingly, we can see from (\ref{eq:result1}) and \ref{eq:result2}) that the likelihood, conditional on incomplete past data with $L$ missing data points, has  a nice structure. It happens to be a combination of different weighted kernel density estimates whose bandwidths range over $\underline{h}$, $\sqrt{2}\underline{h}, \ldots,\sqrt{L+1} \underline{h}$, where the kernel density estimate with bandwidth $\sqrt{m+1} \underline{h}$ is based on data available before the $m$th missing time-point, and its weights are computed based on products of terms depending on $m-$tuples from the time-points corresponding to missing data. 
 \begin{rem}
 \label{rem1}
 To model incomplete temporal data,  the following computationally simpler weighted kernel density could be considered, which is same as using only the first summation term in our likelihood (\ref{eq:result2}) after renormalisation.
  \begin{eqnarray}
  \tilde{f}_{\Theta}\left(s_{n+1}=s\vert s_{(1:n)\backslash \mathcal{S}_{L+1}}\right) &=& \sum_{i\in(1:n)\backslash\mathcal{S}_{L+1}}\tilde{w}_\theta(i,n+1)\cdot \kappa_2(s-s_i,\underline{h}) , \label{eq:PR}\\  
 \mbox{where } \tilde{w}_\theta(i,n+1) &=& \frac{w_\theta(i,n+1)}{\sum_{j\in(1:n)\backslash\mathcal{S}_{L+1}}w_\theta(j,n+1)} . 
 \end{eqnarray}
 \citet{porter_reich2012} use such a formulation for modeling the next crime event location based on crime series of different offenders.  However, we note that for incomplete temporal data (i.e. when $\mathcal{S}_{L+1}$ is non-empty), this would not be a stochastically consistent model specification. For example, in particular, it is easy to check that 
 \begin{eqnarray}
 \tilde{f}_{\Theta}\left(s_{n+1}\vert s_{(1:n-1)}\right)\ne \int  \tilde{f}_{\Theta}\left(s_{n+1}\vert s_{(1:n)}\right)\cdot \tilde{f}_{\Theta}\left(s_{n}\vert s_{(1:n-1)}\right) ds_n. \end{eqnarray}
In our formulation,  we give a stochastically consistent specification by starting with the model (\ref{eq:fullmod}) and then obtaining the implied likelihood (\ref{eq:result1}) given incomplete past data. In Appendix Section \ref{S:simulation}, we illustrate using a simulation study that using the model  (\ref{eq:PR}) instead of (\ref{eq:result1}) for modeling incomplete temporal data can be inadequate for estimating the parameters.
 \end{rem}
 
 \subsection{Expert-prior}
\label{S:expert}
Based on the second author's extensive interactions with several police officers co-deployed in anti-naxalite operations, we developed the following methodology for constructing the expert-prior on possible gang-locations. On day $n$, we obtain the probability density, denoted by $E^{(n}(s)$, at location $s$, based on expert knowledge, using the following steps. The steps are also illustrated in Figure \ref{F:expert_prior_steps}, for day $n=716$ as an example, for Gang A.
\begin{itemize}
\item[] (i) Identify forests and CRPF camps on the map
\item[] (ii) Mark the locations with high forest density, i.e. $F(s)\geq 50\%$.
\item[] (iii) Unmark all locations among (ii) that are less than 3 km away from any CRPF camp.
\item[] (iv) Identify the last $k_0$ observed gang-locations (we illustrate with $k_0=3$). Retain marked locations from (iii) that belong to the extended convex hull (i.e. convex hull extended by 10 km) around last $k_0$ gang-locations. Unmark all other locations.
\item[] (v) If available intelligence through the informant network is not more than 10 days old, then in addition to what remains marked in (iv), also mark locations within 10 km of the location indicated by such intelligence. 
\end{itemize}
All locations that remain marked at the end of step (v) together form the support of the expert-prior map. Steps (i) to (iii) imply that locations with high forest density and  away from CRPF camps are likely to be preferred by the gang. However, there would be several such locations across the Jharkhand state.  Therefore, the police tend to consider recent knowledge of past gang-locations,  typically past 3 observed locations. Step (iv) formalizes this by focusing on an extended convex hull around last three observed locations provided they have high forest density and are away from CRPF camps. We consider an extended convex hull to allow for fact that gangs could have moved since intelligence was received and hence include a buffer of 10 km around the last 3 observed locations. Additional intelligence, as mentioned in step (v), is incorporated when available.   If the locations based on intelligence and that remain marked at the end of step (iv) overlap, then double the weight is assigned to the overlapping locations to indicate a relatively higher likelihood of the gang's presence. Such an overlap  occurs between intelligence input and the rest of expert-prior  in part (v) of Figure \ref{F:expert_prior_steps}. In part (vi) of the figure, we  also show what turned out to be the actual gang-location.

The above-mentioned process can be mathematically defined as follows. At location $s$, let the forest density be denoted by $F(s)$,  its distance to the closest CRPF camp by $C(s)$, its distance to the location of the first intelligence input (if available) by $I_1(s)$, distance to the location of the second intelligence input (if available) by $I_2(s)$, and $\mathcal{X}_{k_0}$ denote the extended convex hull around recent $k_0$ observed locations. Then, the probability density  based on expert knowledge can be obtained as:
\begin{eqnarray}
&&E^{(n)}(s) \propto E^{(n)}_0(s)=\nonumber\\
&& \begin{cases} 1, ~\mbox{ if }F(s)\geq 0.5, ~C(s)>3 \mbox{ km and } s\in \mathcal{X}_{k_0}, \\
 1,~~~\mbox{if } I_1(s)\leq 10 \mbox{ km},  \mbox{ but }F(s)< 0.5 \mbox{ or } C(s)\leq 3 \mbox{ km or } s\notin \mathcal{X}_{k_0}, \\
 2,~~~\mbox{if both } I_1(s), I_2(s)\leq 10 \mbox{ km}   \mbox{ but }F(s)< 0.5 \mbox{ or } C(s)\leq 3 \mbox{ km or } s\notin \mathcal{X}_{k_0}, \\
 2,~~~\mbox{ if }F(s)\geq 0.5 \mbox{ and }C(s)>3 \mbox{ km and } s\in \mathcal{X}_{k_0}\mbox{ and }I_1(s)\leq 10 \mbox{ km},\\
 3,~~~\mbox{ if }F(s)\geq 0.5,~C(s)>3 \mbox{ km and } s\in \mathcal{X}_{k_0}\mbox{ and both }I_1(s),~I_2(s)\leq 10 \mbox{ km},\\
 0,~~~\mbox{ for any other condition}.
 \end{cases}\label{eq:heatmap0}
\end{eqnarray}
Although equation (\ref{eq:heatmap0}) assumes that at most two different intelligence inputs may be available on any given day, it can be easily extended to similarly incorporate more than two intelligence inputs. To obtain the probability density $E^{(n)}(s)$ we need to normalize as follows:
\begin{eqnarray}
&&E^{(n)}(s) = \frac{E^{(n)}_0(s)}{\int E^{(n)}_0(s)ds}.\label{eq:heatmap}
\end{eqnarray}
 \begin{figure}[ht]
\center
\caption{\label{F:expert_prior_steps} Illustration of steps for constructing the expert-prior $E^{(n)}(s)$ for day $n=716$. (vi) shows the actual gang-location on that day for reference.}
\begin{subfigure}{.45\textwidth}
\center
\caption{Identify forest and CRPF camps}
\includegraphics[width=6cm, height=6cm, angle=0]{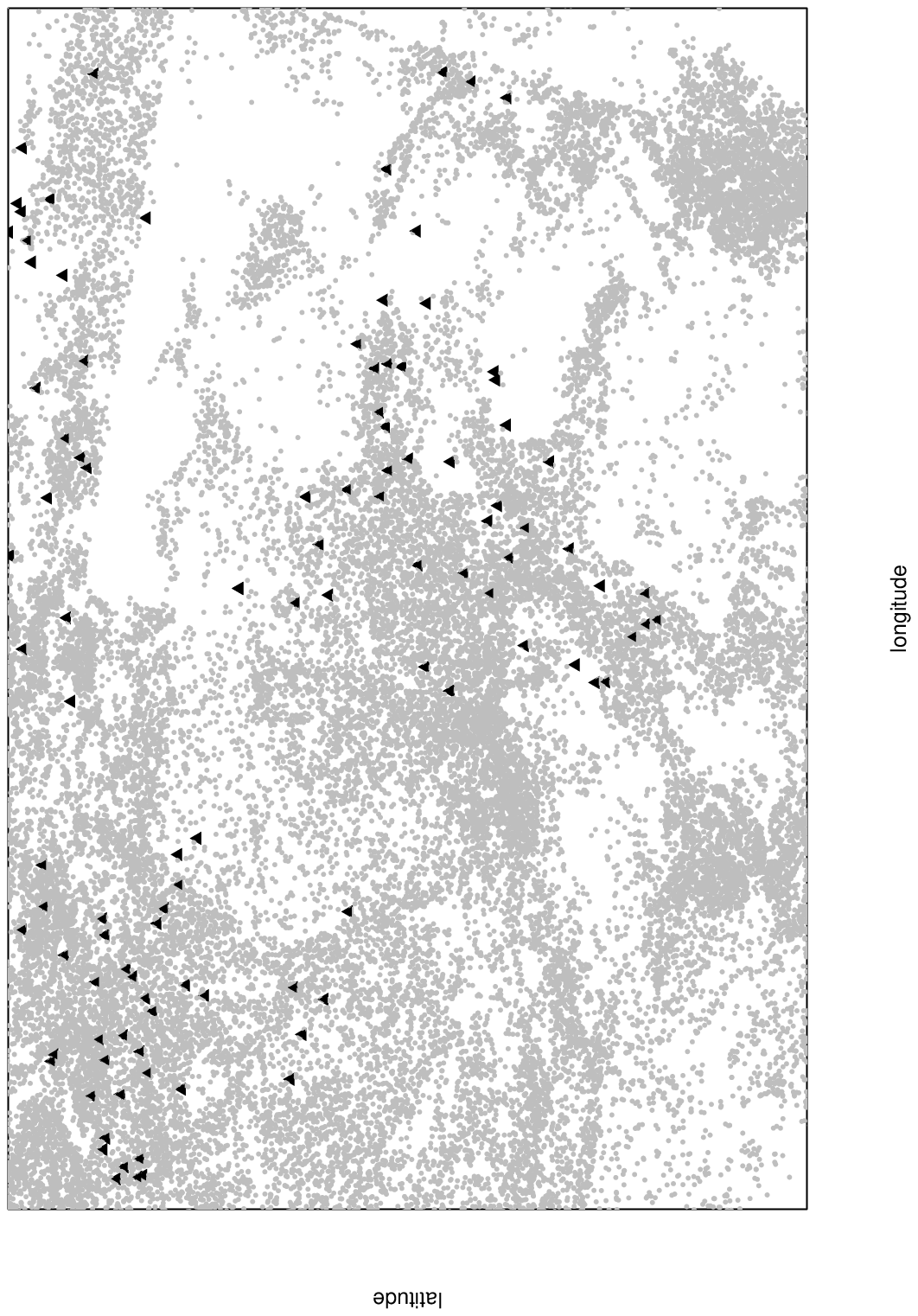}
\end{subfigure}
~\begin{subfigure}{.45\textwidth}
\center
\caption{Mark high forest density locations}
\includegraphics[width=6cm, height=6cm, angle=0]{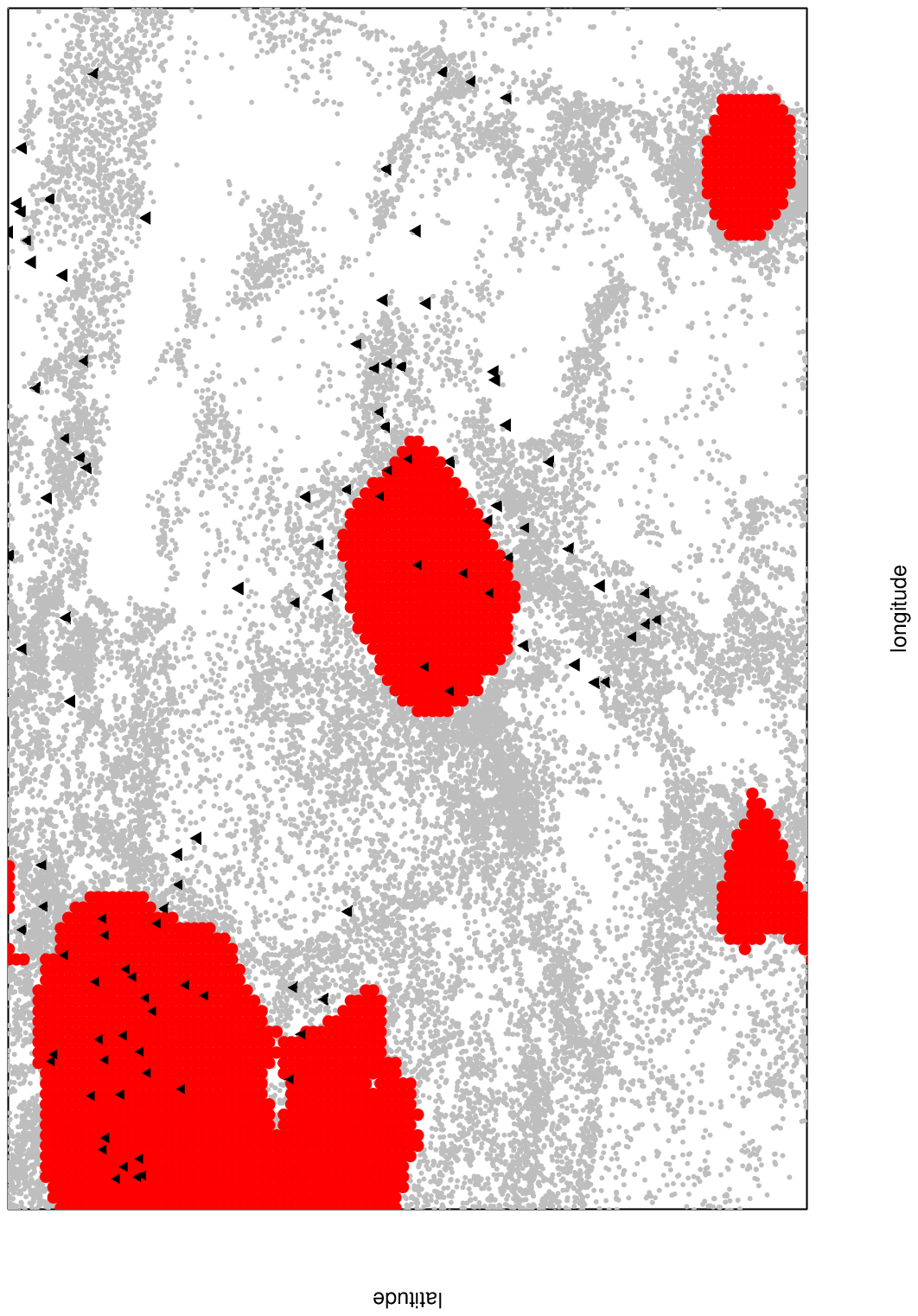}
\end{subfigure}

\begin{subfigure}{.45\textwidth}
\center
\caption{Unmark locations close to CRPF camps}
\includegraphics[width=6cm, height=6cm, angle=0]{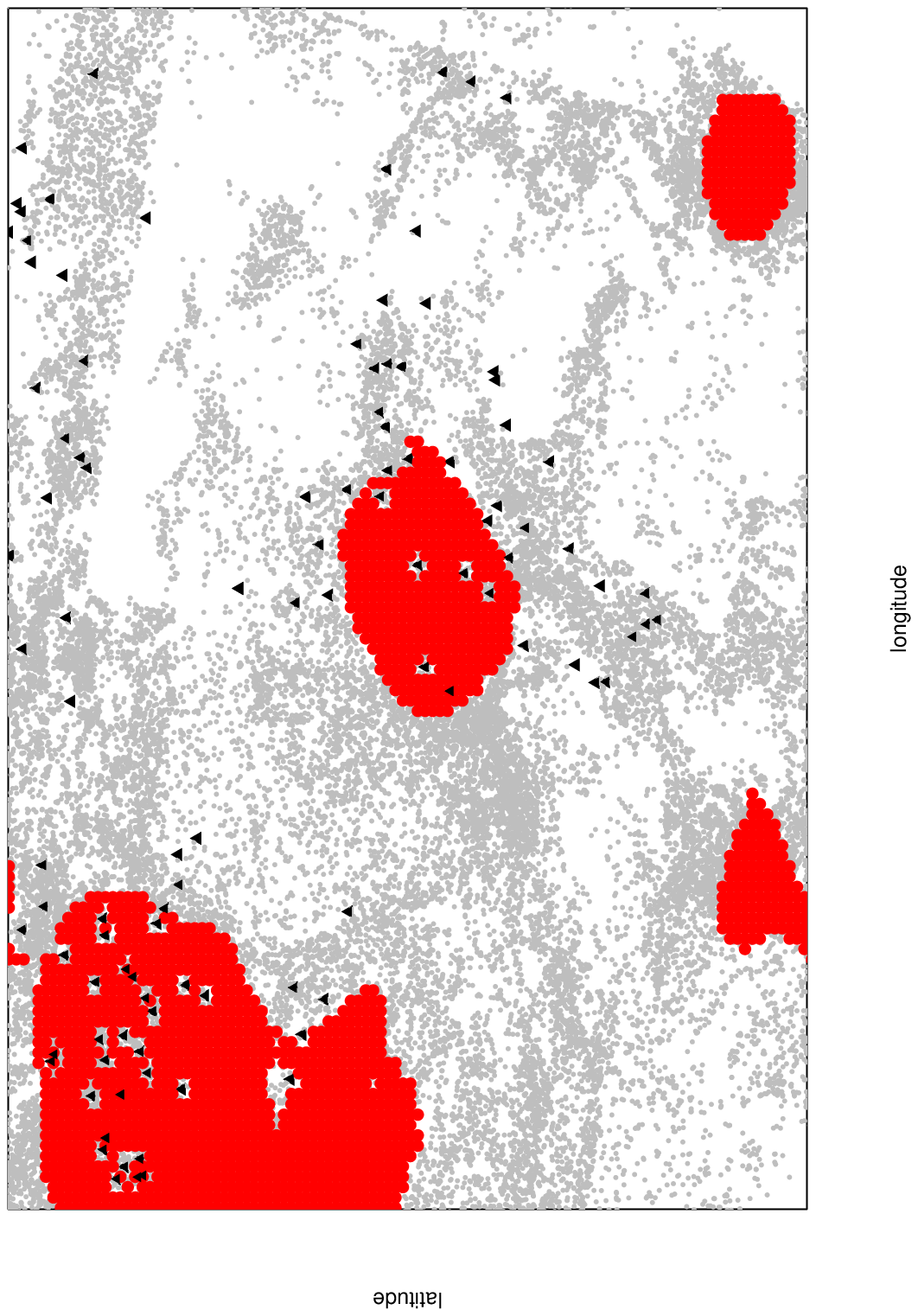}
\end{subfigure}
~\begin{subfigure}{.45\textwidth}
\center
\caption{Retain only locations within extended convex hull of past 3 locations}
\includegraphics[width=6cm, height=6cm, angle=0]{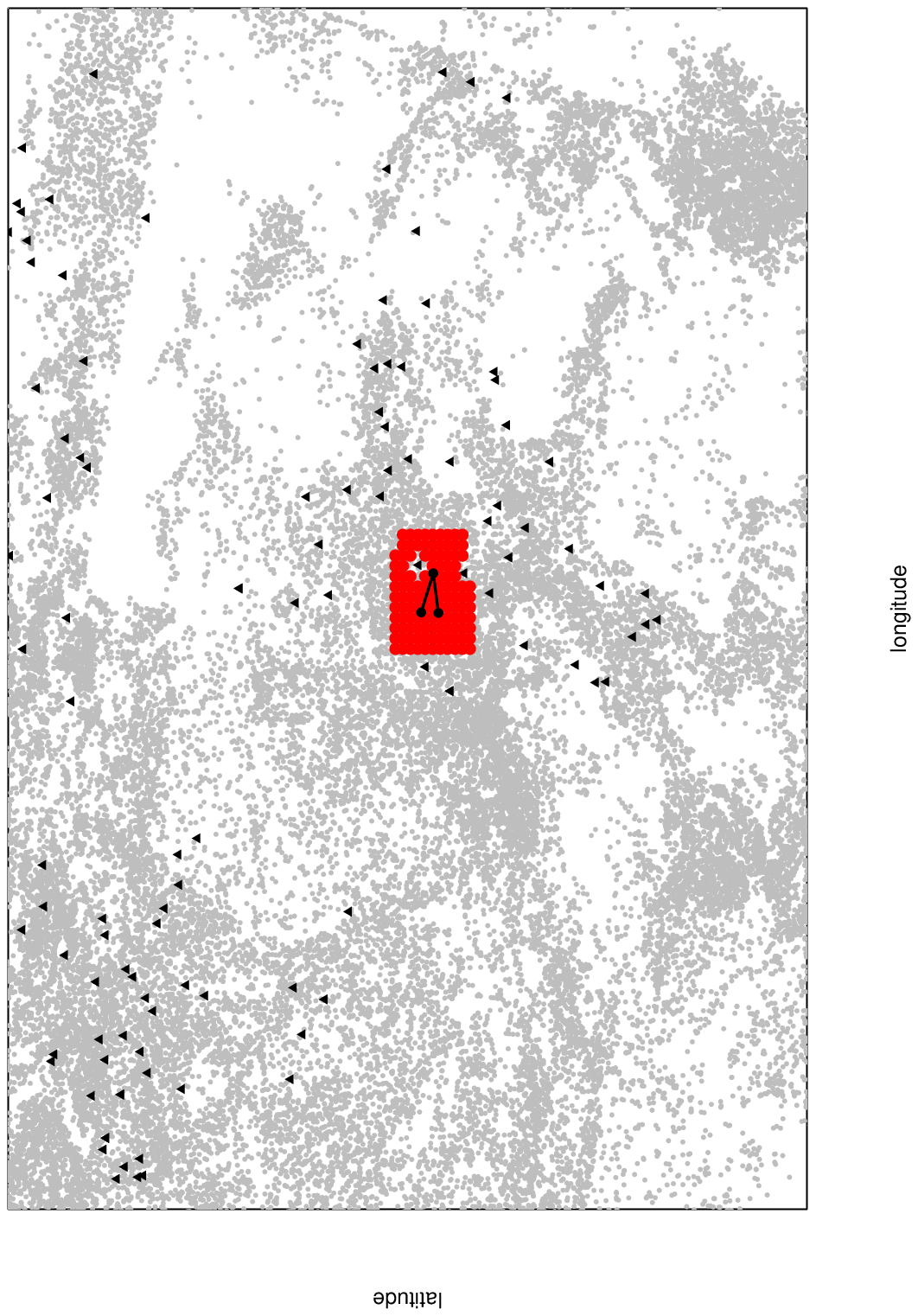}
\end{subfigure}

\begin{subfigure}{.45\textwidth}
\center
\caption{Mark additional locations based on informant intelligence}
\includegraphics[width=6cm, height=6cm, angle=0]{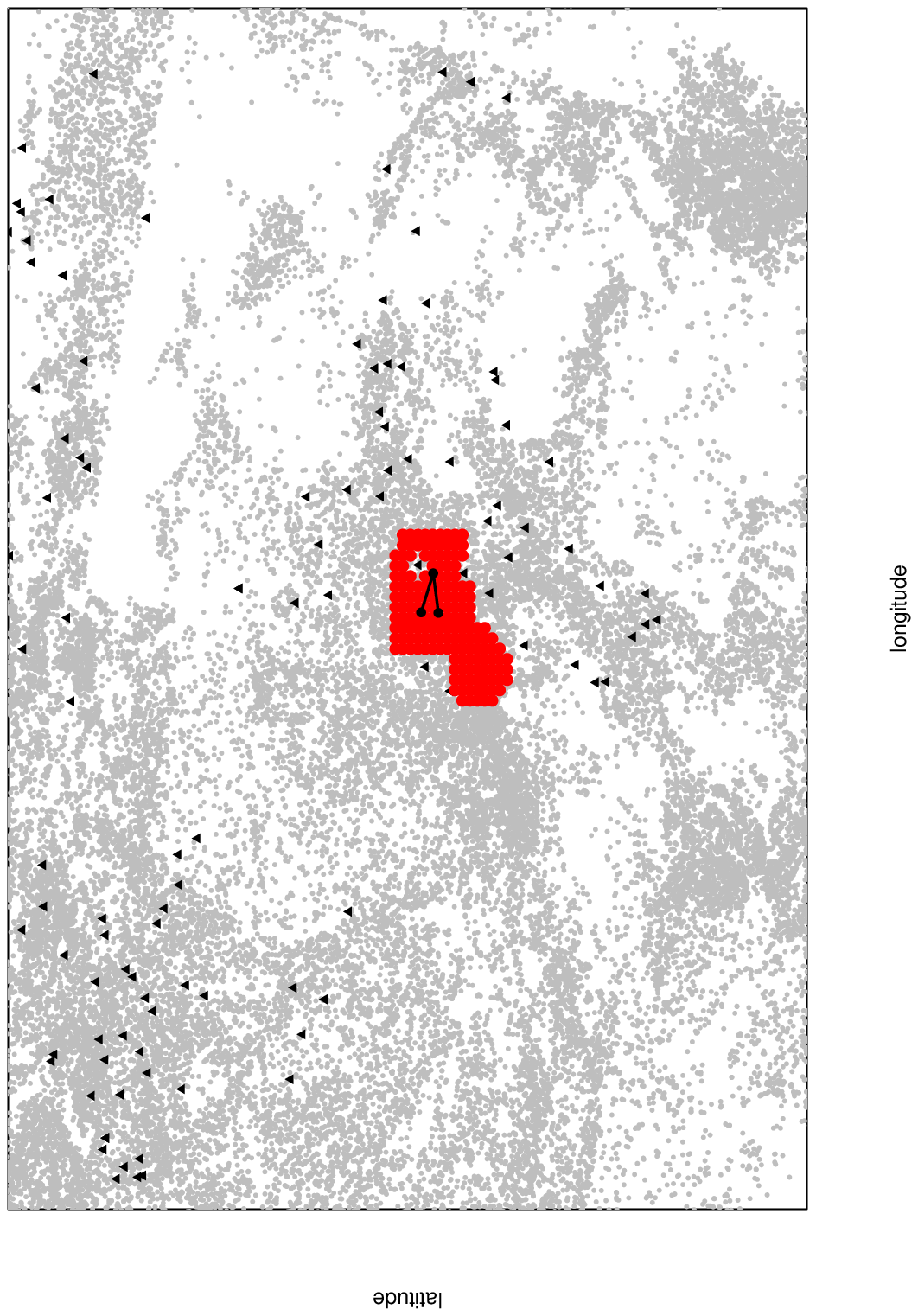}
\end{subfigure}
~\begin{subfigure}{.45\textwidth}
\center
\caption{Expert-prior map for the day along with the observed gang-location (blue dot) }
\includegraphics[width=6cm, height=6cm, angle=0]{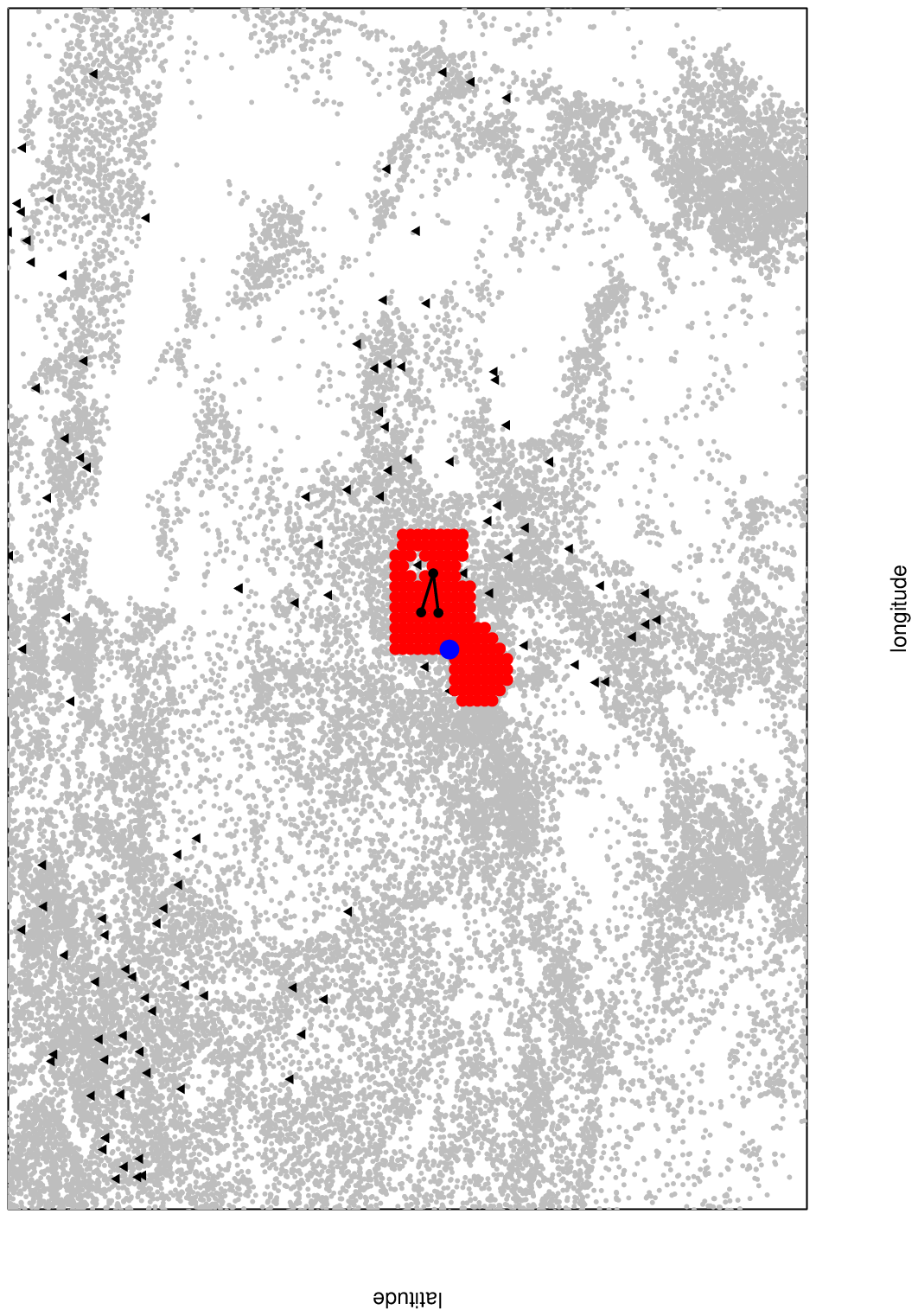}
\end{subfigure}
\end{figure}
We highight that the expert prior is stated in terms of the possible gang-locations, and not in terms of the model parameters $\Theta$, as  required in a  standard Bayesian analysis. Further, our context requires that on any given day, we combine two priors: (a)  the expert-prior map constructed for that particular day, (b) prior on the model parameters $\Theta$ that has been obtained through posterior distribution updates of the data and priors from the past. In the next section, we device an approach to update the posterior considering these requirements.
 \section{Estimation and prediction} 
 \label{S:est_pred}
 Here, we give our procedure for sequential Bayesian estimation for the model by incorporating the expert-prior developed in Section \ref{S:model_expert}, and  for obtaining the predictive density of the gang-location on any given day $n$. We note that the procedure can be applied sequentially on data that becomes available on multiple gangs on any given day. We assume that the parameters of the model $\Theta=(\theta, h)$ are common for all gangs, but for any day $n$, the locations $s_{1:n}$, likelihood $f_{\Theta}\left(s \vert \mathcal{D}_n \right)$ and expert-map $E^{(n)}(s)$ depend on the particular gang $g$ being considered while updating the posterior. It is possible that on a given day, location for multiple gangs may be observed. In that case, we just repeatedly run the Bayesian update procedure described below, by considering one gang at a time on that day. Predictive density is computed and predictive assessment is carried out for each gang. With this understanding, in the interest of notational simplicity, we suppress the dependence of the past data, likelihood and expert-map on the specific gang $g$.
\subsection{Sequential Bayesian Estimation} 
 \label{S:estimation}
 On any given day $n$, we will have the following known to us,
 \begin{itemize}
\item[](i) the likelihood function for the gang-location on the $n$th day given the past observed locations, i.e. $f_\Theta(s\vert \mathcal{D}_{n-1})$ (as in equations \ref{eq:proposed_model}, \ref{eq:result1}),
 \item[](ii) a prior $\Pi^{(n)}_d(\Theta), ~~\Theta\in \mathcal{S}_\Theta$, obtained from the recent previous Bayesian update,
 \item[](iii) an expert-prior map $E^{(n)}(s)$ for possible locations of the gang on the given day ($n$), as discussed in Section \ref{S:expert}, along with a credibility weight $p^{(n)}$, which is the belief over expert knowledge. 
\end{itemize}
We note that in practice, the expert map $E^{(n)}$ and the corresponding credibility weight $p^{(n)}$ will be subjective. What we suggest in this paper is a reasonable recommendation based on  the experience of the second author's discussions with police officers co-deployed in anti-naxalite operations, and can be taken as a starting point for future enhancements. In Section \ref{S:expert}, we provided systematic approach to obtain the expert map $E^{(n)}$. Further, we use the following for $p^{(n)}$.
\begin{eqnarray}
p^{(n)}= \begin{cases}0.5,\mbox{ if informant intelligence is available during prior construction},\\
0.1,\mbox{ if informant intelligence is not available.}\end{cases} \label{eq:p_n}
\end{eqnarray}
The expert-prior is stated in terms observables,  i.e. the possible gang-locations, and not stated in terms of the model parameters $\Theta$. However, for a standard Bayesian analysis, we need the prior in terms of model parameters in order to subsequently derive the posterior inference. One approach we explored to addressing this issue was to interpret $E^{(n)}(s)$ as the marginal distribution based on the likelihood (\ref{eq:proposed_model}), and  following \citealt{efron2014} to derive the desired prior $\Pi^{(n)}(\Theta)$  as a solution to the following equation, 
\begin{eqnarray}
\int  f_\Theta\left(s \vert\mathcal{D}_n\right) \Pi^{(n)}(\Theta) d\Theta = E^{(n)}(s). \label{eq:prior_lik_marg}
\end{eqnarray}
It is not reasonable to expect that  the equation (\ref{eq:prior_lik_marg}) will be feasible in practice. Therefore, to logically incorporate the expert-map as part of the prior on parameters, we device a different approach. We first extend the parameter set by including an additional dummy value $\Theta_0$ as 
 \begin{eqnarray}
 \tilde{\mathcal{S}}_{\Theta} = \mathcal{S}_{\Theta} \cup \{\Theta_0\}, \label{eq:ext_par}
\end{eqnarray}
and define the  extended likelihood $\tilde{f}_{\Theta}$, as
\begin{eqnarray}
\tilde{f}_{\Theta}\left(s_{n}\vert \mathcal{D}_{n-1}\right) =\begin{cases} f_{\Theta}\left(s_{n}\vert \mathcal{D}_{n-1}\right), \mbox{ if } \Theta\in \mathcal{S}_\Theta \\
E^{(n)}(s)~~~~~~~~, \mbox{ if }\Theta=\Theta_0.\end{cases} \label{eq:ext_lik}
\end{eqnarray} 
 Since we need to give a credibility weight of $p^{(n)}$ on the expert map, the corresponding prior on the extended parameter set can be written as
 \begin{eqnarray}
\tilde{\Pi}^{(n)}(\Theta)= \begin{cases}(1-p^{(n)})\cdot\Pi^{(n)}_d(\Theta), ~~\Theta\in \mathcal{S}_\Theta, \\
p^{(n)}~~~~~~~~~~~~~~~~~~~~, ~~\Theta=\Theta_0. \end{cases} \label{eq:ext_prior}
\end{eqnarray}
 The posterior distribution on the extended parameter set, i.e. $\Theta\in \tilde{\mathcal{S}}_\Theta $, based on the likelihood  (\ref{eq:ext_lik}) and prior  (\ref{eq:ext_prior}),  is given by
 \begin{eqnarray}
\tilde{\Pi}^{(n)}(\Theta\vert \mathcal{D}_{n})\propto&& \begin{cases} \tilde{f}_\Theta(s_{n}\vert\mathcal{D}_{n-1})\cdot \tilde{\Pi}^{(n)}(\Theta),\mbox{ if data } s_n \mbox{ is observed}, \\
\tilde{\Pi}^{(n)}(\Theta)~~~~~~~~~~~~~~~~~~, \mbox{ if data } s_n \mbox{ is not observed}.\end{cases} \label{eq:posterior_ext}
\end{eqnarray}
Note that  equation (\ref{eq:posterior_ext}) explicitly recognizes that if  locations for none of the gangs is observed on day $n$, then posterior update for $\Theta$ is not possible, and hence the latest known prior on $\Theta$ will continue to prevail.
Before proceeding to repeat the  prediction and estimation for day $(n+1)$, we will require $\Pi^{(n+1)}_d(\Theta)$, which is given by
\begin{eqnarray}
\Pi^{(n+1)}_d(\Theta)= \frac{ \tilde{\Pi}^{(n)}(\Theta\vert \mathcal{D}_{n})}{1-\tilde{\Pi}^{(n)}(\Theta_0\vert \mathcal{D}_{n})},  ~~\Theta \in \mathcal{S}_\Theta. \label{eq:pi_n_d}
\end{eqnarray} 
 Equation (\ref{eq:pi_n_d}) is  the posterior distribution (\ref{eq:posterior_ext}) restricted to the parameters of the main likelihood, i.e. $\Theta \in \mathcal{S}_\Theta $.  In our problem, we would require a mechanism to update our belief of the parameters sequentially as and when new data on gang-locations become available over different days. Particle filter algorithms are particularly useful for carrying out sequential updates, see \citet{carvalhoetal2010a} and \citet{Daietal2022} for review of particle learning and other sequential updating procedures.  
 Also, while \citet{carvalhoetal2010b} consider sequential Bayesian updates in the context of general mixtures, developing a particle learning algorithm on those lines with our likelihood (\ref{eq:result1}) does not appear straightforward,  and we defer that exploration for a future paper. For our work, it suffices to use a basic particle filtering (PF) approach combined with kernel smoothing and some computational enhancements, which we will describe below. Given that our expert-prior construction depends on the last 3  observed gang-locations, our update starts on day $n=n_0$, which is after the third data point has been observed for at least one of the gangs. The steps of the PF algorithm, after fixing a large value of $N$,  are as follows:
\begin{itemize}
\item[] Step 0. For $n=n_0$, let $\tilde{\Pi}_d^{(n)}$ be a non-informative flat prior on a reasonably large compact support for the parameters. Draw  $\Theta_1, \Theta_2, \ldots, ~\Theta_N \stackrel{iid}{\sim} \tilde{\Pi}_d^{(n)}.$
\item[] Step 1. Draw $\tilde{\Theta}_1, \tilde{\Theta}_2, \ldots, ~\tilde{\Theta}_N$ as iid from a distribution supported on $\Theta_0, \Theta_1, \ldots, \Theta_N$ such that $$P(\tilde{\Theta}_i= \tilde{\Theta}) =\begin{cases}p^{(n)}, ~\mbox{ if }\tilde{\Theta}=\Theta_0,\\ \frac{(1-p^{(n)})}{N}, \mbox{ if }\tilde{\Theta}=\Theta_j,~ j=1,2,\ldots,N.   \end{cases}.$$
\item[] Step 2. Approximate $\tilde{\Pi}^{(n)}(\Theta\vert \mathcal{D}_n)$ with the discrete distribution which can take values $\tilde{\Theta}_{1},  \tilde{\Theta}_{2}, \ldots, \tilde{\Theta}_{N}$ with probablities $p_1, p_2, \ldots, p_N$, where 
\begin{eqnarray}
p_i=\begin{cases} \frac{\tilde{f}_{\tilde{\Theta}_i}(s_n\vert \mathcal{D}_{n-1})}{\sum_{i=1}^N \tilde{f}_{\tilde{\Theta}_i}(s_n\vert \mathcal{D}_{n-1})}, \mbox{ if location } s_n \mbox{ on day }n\mbox{ is observed},\\
\frac{1}{N}~~~~~~~~~~~~~~~~~~,\mbox{ if location } s_n \mbox{ on day }n\mbox{ is not observed}.\end{cases}\label{eq:PFstep1}
\end{eqnarray}
We note that some of the values in $\tilde{\Theta}_{1},  \tilde{\Theta}_{2}, \ldots, \tilde{\Theta}_{N}$ may be repeated. Let $\tilde{\Theta}_{00}$,$  \tilde{\Theta}_{01}$, $\tilde{\Theta}_{02}, \ldots, \tilde{\Theta}_{0N_0}$ denote the distinct values, where $\tilde{\Theta}_{00}=\Theta_0$, $N_0\leq N$. Then  $\tilde{\Pi}^{(n)}(\Theta\vert \mathcal{D}_n)$ can be approximated by the discrete distribution given by 
\begin{eqnarray}
&&\begin{bmatrix}\mbox{Possible values} &\tilde{\Theta}_{00} & \tilde{\Theta}_{01} &\tilde{\Theta}_{02}&\ldots & \tilde{\Theta}_{0N_0}\\ \mbox{Probabilities}&q_0 & q_1 &q_2 & \ldots &q_{N_0}. \end{bmatrix},~\label{eq:PFstep2}\\
\mbox{where}&&\nonumber\\
&&q_j = \sum_{i\in\{1, \ldots,N\}:\tilde{\Theta}_i=\tilde{\Theta}_{0j}}p_i, ~j=1,2,\ldots, N_0\label{eq:PFstep2q}
\end{eqnarray}
\item[] Step 3. Our last step for day $n$ is to obtain a sample $\Theta_1, \Theta_2, \ldots, \Theta_N$ from $\Pi_d^{(n+1)}(\Theta)$ as in (\ref{eq:pi_n_d}), which is essentially $\tilde{\Pi}^{(n)}(\Theta\vert \mathcal{D}_n)$ obtained in the previous step but restricted to only values of $\Theta \ne \Theta_0$.
\item[] Step 4.  Set $n\rightarrow n+1$ and repeat steps 1,2,3.
\end{itemize}

In our implementation, we apply a few computational enhancements to the above steps. First, while step 3 can be implemented as described above, one of the known issues that can arise, if the values of the particles are restricted to a fixed set of values, is degeneracy (see \citealt{carvalhoetal2010a}), where the posterior distribution of a parameter may start concentrating on a single value (particle).  To avoid this, we use a kernel density approximation for $\Pi_d^{(n+1)}(\Theta)$ before drawing the new set of particles for the next stage. The possible parameter values in Step 3 are $\Theta_{0j}=(\theta_{0j},h_{0j})$ for $j=1,2,\ldots, N_0$.
We note that the possible values of each parameter are positive  and assume that each is contained within a (large enough) compact interval, i.e. $h_{0j}\in[h_{min}, h_{max}]$ and $\theta_{0j}\in[\theta_{min}, \theta_{max}]$. Since beta-kernel estimators are a natural choice for estimating probability densities with compact support (e.g. \citealt{chen1999}), we use this estimator to approximate $\Pi_d^{(n+1)}(\Theta)$.  Accordingly, we first define  transformed variables 
$$h_{0j}^\prime= \frac{h_{0j}-h_{min}}{h_{max}-h_{min}} \mbox{ and }\theta_{0j}^\prime= \frac{\theta_{0j}-\theta_{min}}{\theta_{max}-\theta_{min}}.$$
Then, we obtain the beta-kernel density estimate for the parameters as
\begin{eqnarray}
G(h^\prime,\theta^\prime) = \frac{1}{N_0}\sum_{j=1}^{N_0} \frac{q_j}{1-q_0}\cdot \phi_{a_i,b_i}(h^\prime)\cdot\phi_{c_i,d_i}(\theta^\prime), \label{eq:kern_theta}
\end{eqnarray}
 where $\phi_{a_i,b_i}(h^\prime)$ and $\phi_{c_i,d_i}(\theta^\prime)$ are beta distributions  in variables $h^\prime$ and $\theta^\prime$ with $(a_i, b_i)$ and $(c_i, d_i)$ denoting the respective shape parameters of the beta densities. We compute $a_i$ and $b_i$ by matching the mean and standard deviation of $\{h^\prime_{oj} \}$ weighted by $\{\frac{q_j}{1-q_0}\}$, $j=1,2,\ldots, N_0$ , and similarly for $(c_i,d_i)$. Then we draw a sample $\Theta_1, \Theta_2,\ldots, \Theta_{N}$ using (\ref{eq:kern_theta}) with appropriate rescaling. Second, to speed up the computations, we approximate the values of the so drawn $\Theta_1, \Theta_2,\ldots, \Theta_{N}$, to the nearest deciles in the distribution. This simplifies the computational effort significantly as we would then need to compute $p_i$ in Step 2 only at distinct values of the decile-approximated $\Theta_i$, and just multiply with their frequency of occurrence. 
\subsection{Predictive density} 
\label{S:pred_density}
We can use the model as estimated in the previous section to generate predicted density at any location $s$, for day $n$, conditional on the past data available until day $(n-1)$.  The predictive density for the location of a gang on day $n$, of course before observing the gang's location $s_n$ on that day,  can be written as
\begin{eqnarray}
\widehat{f}(s \vert \mathcal{D}_{n-1}) &=& \int \tilde{f}_\Theta(s|\mathcal{D}_{n-1})\cdot\tilde{\Pi}^{(n)}(\Theta) d\Theta \label{eq:predictive0}\\
 &=&(1-p^{(n)})\cdot \int f_\Theta(s|\mathcal{D}_{n-1})\cdot\Pi^{(n)}_d(\Theta)d\Theta + p^{(n)}\cdot E^{(n)}(s).\label{eq:predictive}
\end{eqnarray}
We regret the slight abuse of notation in the interest of keeping equation (\ref{eq:predictive0}) elegant. $\tilde{\Pi}^{(n)}(\Theta)$ is strictly not a pdf on $\tilde{\mathcal{S}}_\Theta$ as it has point mass at $\Theta_0$, but the meaning is  clear in (\ref{eq:predictive}). As a practical matter, we compute the predictive density (\ref{eq:predictive}) on a finite grid of locations. We cover the state of Jharkhand, India with a grid of boxes, each of dimension 2.5 km $\times$ 2.5 km.  Then, we compute the predictive density for the gang's presence at the centers of the these grid boxes, which we denote by
 \begin{eqnarray}
\mathcal{S}_{grid} = \{s_{01}, s_{02}, \ldots, s_{0p} \}. \label{eq:space_set}
\end{eqnarray}
For any day $n$, our aim will be to mark locations with high predicted density, so as to guide the police on  where patrolling should be focused, and  which of the CRPF camps may need to be alerted. 

\subsection{Predictive assessment} 
\label{S:pred_assess}
 We assess the predictive performance of the model in its ability to predict the actual observed gang-location $s_n$ on any given day $n$, by using the following two metrics.

~\\
(1) {\bf Required Area to be Monitored (RAM):} It is the cumulative area of the locations, in descending order of the predictive probability density, which have to be monitored to capture the actual next location. To assess how well the model predicts the actual gang-location $s_n$ on day $n$, the RAM  is defined as the total area of  all the grid boxes with their centers in the set 
 \begin{equation}
 \mathcal{M}= \left\{ s\in \mathcal{S}_{grid} : ~\widehat{f}(s \vert \mathcal{D}_{n-1})\geq  \hat{f}\left(s_n\vert t \right)\right\} , \label{eq:M}
 \end{equation}
 where the predictive density $\hat{f}$ is as in equation (\ref{eq:predictive}). 

~\\
 Smaller the value of RAM,  better is the model performance. Since RAM does not consider the proximity of  the other locations highlighted for monitoring, to the actual gang-location, we complement it by proposing a second metric, which we call `Area Under the Proximity Curve'.

~\\
(2){\bf Area Under the Proximity Curve (AUPC):}  First, we define the proximity curve, where the x-axis has the different possible values for the percentage area ($p$) that can be monitored, e.g. top 10\%, top 20\% etc. Correspondingly, the y-axis is the minimum of the distances ($m(p)$) of the grid boxes within the given area to be monitored, to the actual gang location. To be precise,
\begin{eqnarray}
&&m(p) = Min\left\{ \mbox{ distance of } s^\prime\mbox{ to }s_n, ~\forall~s^\prime \in \mathcal{T}_p   \right\}, \mbox{ where }\nonumber\\
&&\begin{split}
\mathcal{T}_p&= \left\{ \mbox{ Top 100}p\% \mbox{ grid boxes with center } s^\prime,\right.\\
& \left. \mbox{ in descending order of } \hat{f}\left(s^\prime\vert \mathcal{D}_{n-1}\right)\right\}.
\end{split}
\end{eqnarray}
To assess how well the model predicts the actual gang-location $s_n$ on day $n$, we use the `Area Under the Proximity Curve' (AUPC). The smaller the AUPC better is the model performance. It may be noted that where RAM is high, the model could still be useful if AUPC is low.
\FloatBarrier
\section{Results}
\label{S:results}
 We apply the  model and sequential estimation procedure described in Sections \ref{S:model_expert} and \ref{S:est_pred} to data on three major insurgent gangs (which we label as A, B, C) that operate in the state of Jharkhand, India. In this section, we first discuss the parameter estimates followed by the predictive assessment for the model. Thereafter, we highlight that the importance of using the expert-prior in the model and finally we show how the model could be implemented in practice.
 \subsection{Parameter Estimates}
 The posterior distribution of the model parameters is updated sequentially as and when a location for any of the gangs is observed. There are 96 observed instances, but  on irregularly separated days spanning over 1305 days, for these  three gangs. Figure \ref{F:summary_parA} shows the parameter estimates obtained. For both the parameters, we can see, as expected, that the 95\% credible intervals get narrower as more data become available. However, the number of data points available is still limited in our context, which is reflected in the wide confidence intervals. The value of the posterior mean of $h$ is around 8 km, after the latest update, which can be interpreted as the radius to considered around historical data points. The value of posterior mean for $\theta$ is about 280 days,  after the latest update, which suggests  that one must not just consider the recent few locations of the gang but should weight  a longer history. 

\subsection{Predictive assessment}
Figure \ref{F:GOF1} illustrates the predictive assessment metrics for gang A, viz. RAM and AUPC, which were described in Section \ref{S:pred_assess}. There are 51 observed instances of gang A in our data. The figure shows 48 instances starting from the 4th instance because 3 past instances are required for constructing the expert-prior.  We recall that for either metric, smaller the values better the model prediction. We see that both the metrics generally  improved as more data has become available. Further, RAM for the model is below 1000 sq km in 33 observed instances. In  19 observed instances, the RAM is less than 500 sq km. It is important to emphasise that prediction of the locations of insurgent criminal gangs is a difficult problem as the gangs usually move over very large areas, and data on their locations is not observable very frequently. For example, in our data,   gang A operates in an area of over 5500 sq km, which they can further expand.  One benefit of our model is that it provides a systematic approach to potentially narrow down the focus to a relatively smaller area.  Of course, there could be cases when the model may not help much, as in the 13th instance, where the RAM happens to be over 5000 sq km. Nevertheless,  it is important to note that there are several instances where the focus area is narrowed down to less than 1000 sq km or even 500 sq km. This itself is significant for police operations as it can be used to alert the CRPF camps in the areas highlighted by the model, as well as refine the ``daily task-sheets",  prepared by the police regularly to indicate locations or directions to monitor.  We further note that AUPC can provide additional information to complement RAM. To illustrate this, we highlight (by circling) observed instances 29, 32 and 34 of gang A in both parts (i) and (ii) of Figure \ref{F:GOF1} .  We can see that while the RAM has decreased over these days, the AUPC has increased. A higher RAM but lower AUPC means that the required area to monitor may be more, but some of the other locations marked for monitoring may be close to the actual location. To gain more clarity, Figure \ref{F:proximity-curve_A29_34} shows the proximity curves corresponding to the 29th and 34th observed instances of gang A.  RAM is lower for the 34th instance since the proximity curve touches the x-axis sooner, but its minimum distance of monitored locations to actual location  is much higher. On the other hand, although RAM is higher for the 29th instance, many of the locations marked for monitoring happen to be close, in fact within 20 km of the actual gang location. That this model with such limited observable data can lead the security forces to locations that may be close to the gang's presence is significant in itself, in the context of tracking insurgent gangs as it increases the chance of encounters. 
 
 \subsection{Importance of prior}
Further, our analysis highlights that prior information can be very important in planning police operations. We have provided an approach to construct an expert-prior in Section \ref{S:expert}.  We can see this in Figures \ref{F:with_no_prior_AUPCRAM} (i) and (ii), which highlight (by circling) several instances where the model-with-prior has better RAM and AUPC than the model-without-prior respectively, for gang A. To summarize the comparison for all the 3 gangs, Table \ref{T:GOF_gang} shows the percentage of instances for each of the gangs A, B and C,  when the model-with-prior performed (i) strictly better than, and (ii) at least as good as, the model-without-prior, in terms of RAM as well as AUPC.   We see from part (i) of the table that the AUPC for the model-with-prior turned out to be better more often (i.e. in more than 50\% instances) for all the three gangs.  In terms of RAM, the model-with-prior was better more often for gangs A and B, but not for C.  However, from part (ii) of the table, we note that even in the case of gang C,  in more than 50\% instances the model-with-prior was as at least as good as (if not better than) the model-without-prior in terms of RAM. The prior we constructed is guided by counter-insurgency  tactics adopted by the police. The prior density happens to very coarse, taking constant value for several locations.  Although we do not show the results here, we note that using purely the prior without using the past data, does not yield good results in terms of both RAM and AUPC. However, our findings emphasise that the expert-prior used along with data can play an important role in the context of predicting movement of insurgent gangs. Our prior has potential for improvement with more inputs from the police officers once the model is implemented for use. Our work serves to demonstrate how inputs from the police officers can be usefully integrated with the data, and should motivate the design of improved systems  to capture more relevant information to help refine the priors. Such information may include comprehensive geo-coding of least-effort routes along hilly regions and water bodies, or may also include dynamic information on tell-tale signs of a gang's presence, such as unusual closure of market places, empty streets at unusual times or days in an otherwise busy location etc.

  Figure \ref{F:with_no_prior} shows the prior weight  given to expert knowledge along with the posterior weight obtained after the Bayesian update.  When the posterior weight is higher than the prior weight it can be interpreted as the expert prior turning out to be more important than the credibility it was initially assigned, and vice-versa. Such feedback from the model can be useful to police forces for regularly assessing the importance of the prior vis-a-vis data, and possibly guide improvements in data collection as well as construction of the prior. We note here that minor changes in the cutoffs for the forest density or distance to CRPF camps etc, involved in the expert-prior construction do not lead to drastically different results.

\begin{figure}[ht]
\center
\caption{\label{F:summary_parA} Posterior mean and 95\% credible intervals for the model parameters.   }
\begin{subfigure}{.45\textwidth}
\center
\caption{ $h$ }
\includegraphics[width=7cm, height=7cm, angle=0]{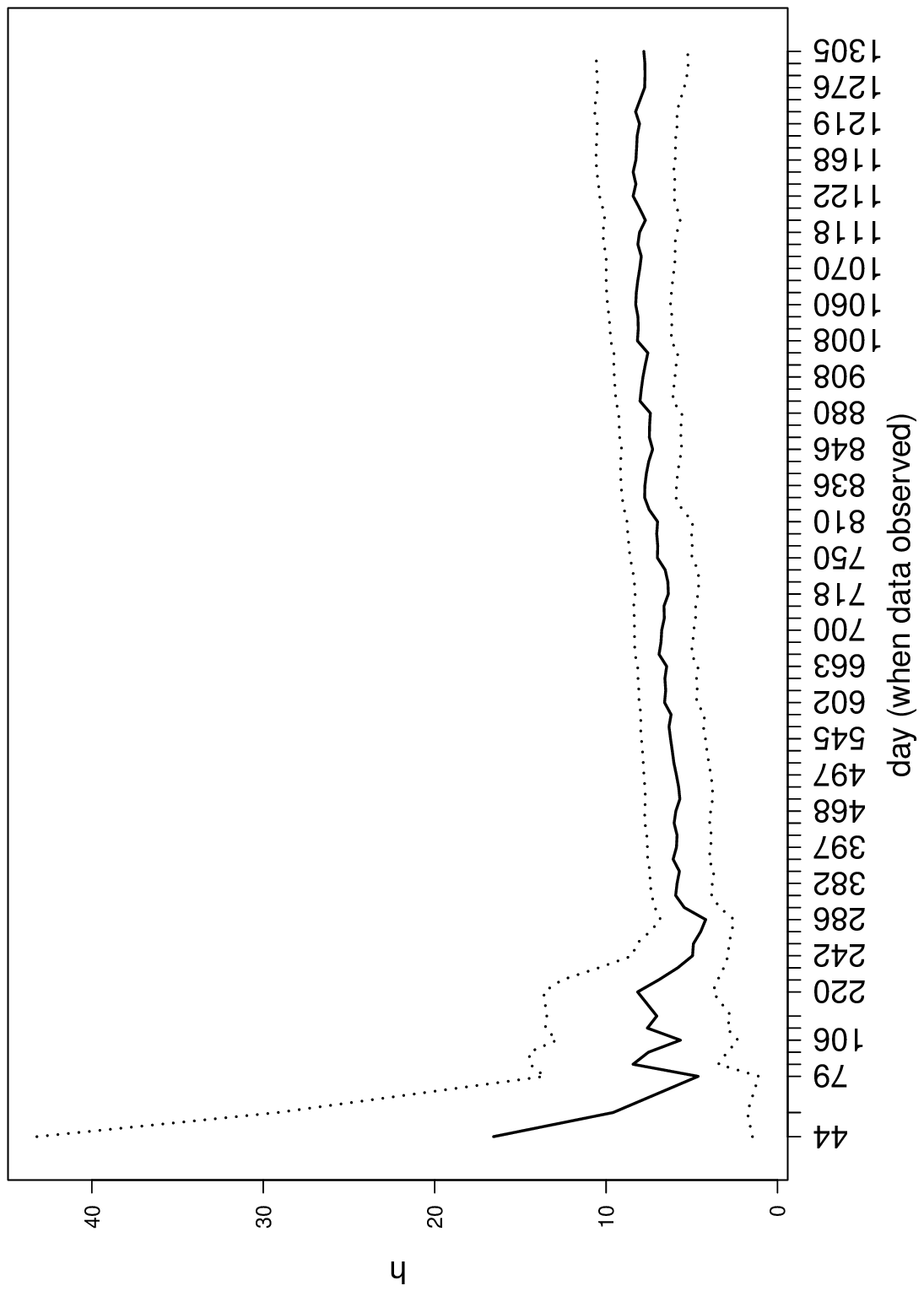}
\end{subfigure}
~\begin{subfigure}{.45\textwidth}
\center
\caption{ $\theta$}
\includegraphics[width=7cm, height=7cm, angle=0]{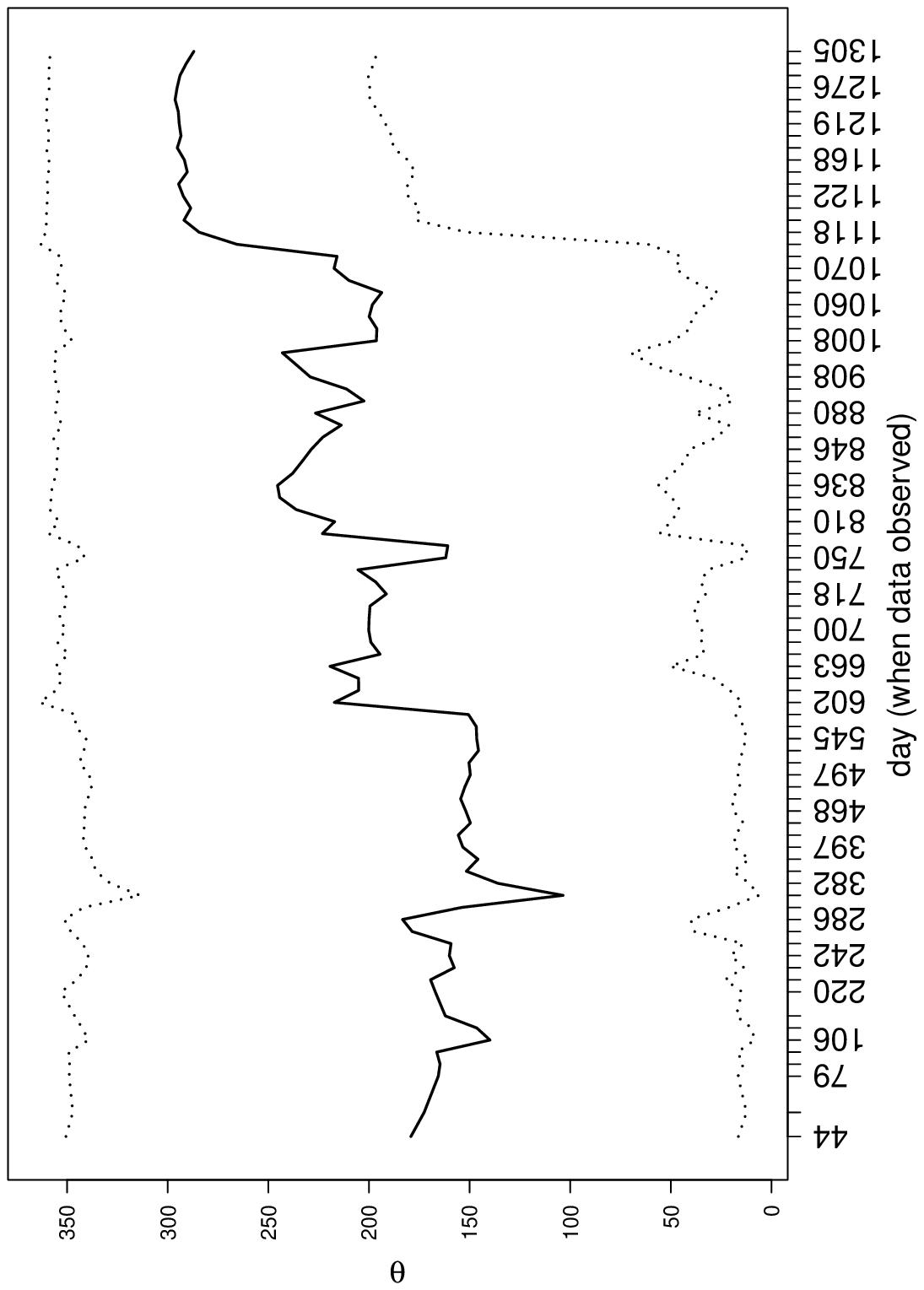}
\end{subfigure}
\end{figure}
 
 \begin{figure}[ht]
\center
\caption{\label{F:GOF1} Predictive assessment metrics for gang A  }
\begin{subfigure}{.45\textwidth}
\center
\caption{ Required Area to Monitor (RAM)}
\includegraphics[width=6.5cm, height=6.5cm, angle=0]{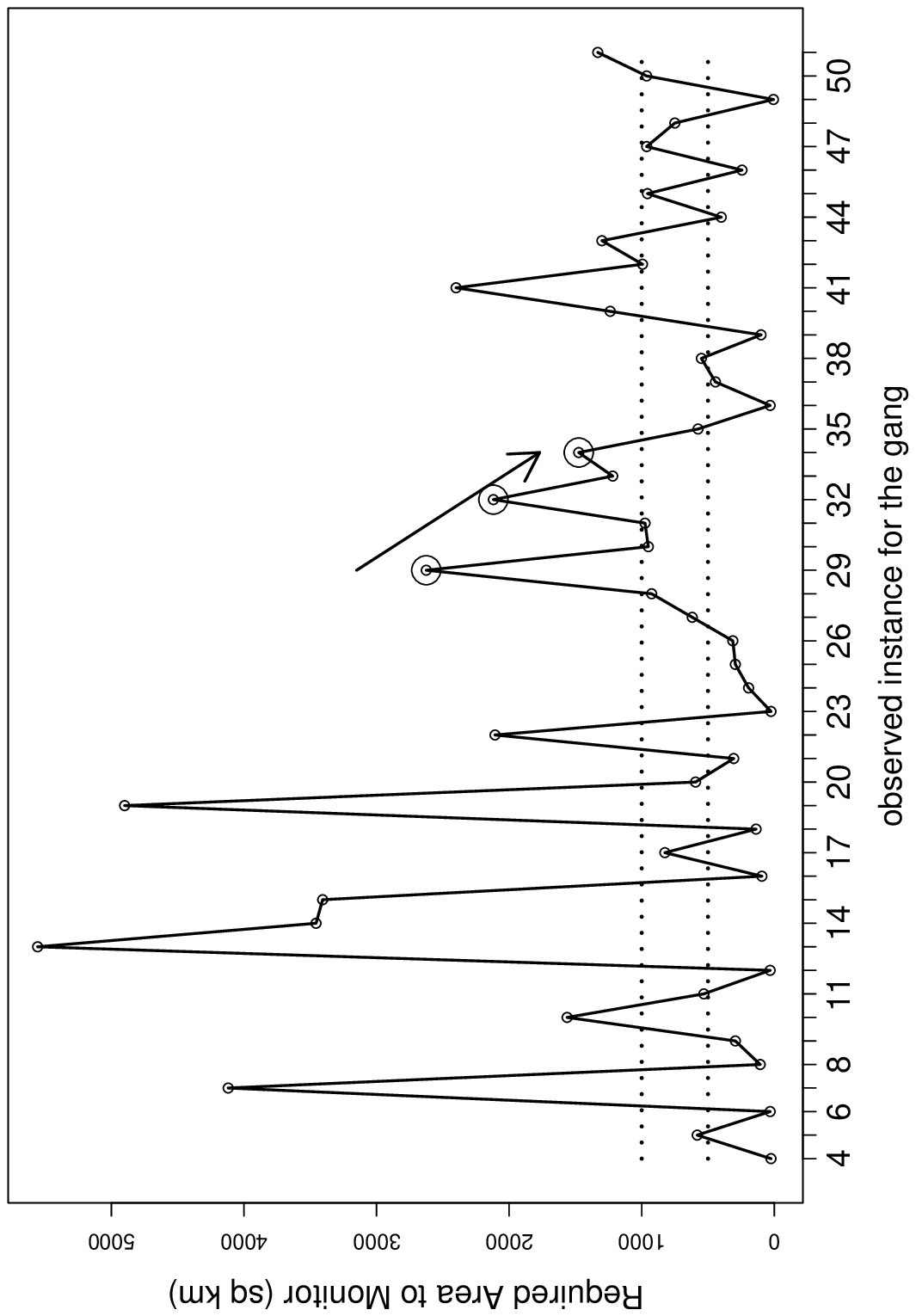}
\end{subfigure}
~\begin{subfigure}{.45\textwidth}
\center
\caption{Area Under Proximity Curve (AUPC)}
\includegraphics[width=6.5cm, height=6.5cm, angle=0]{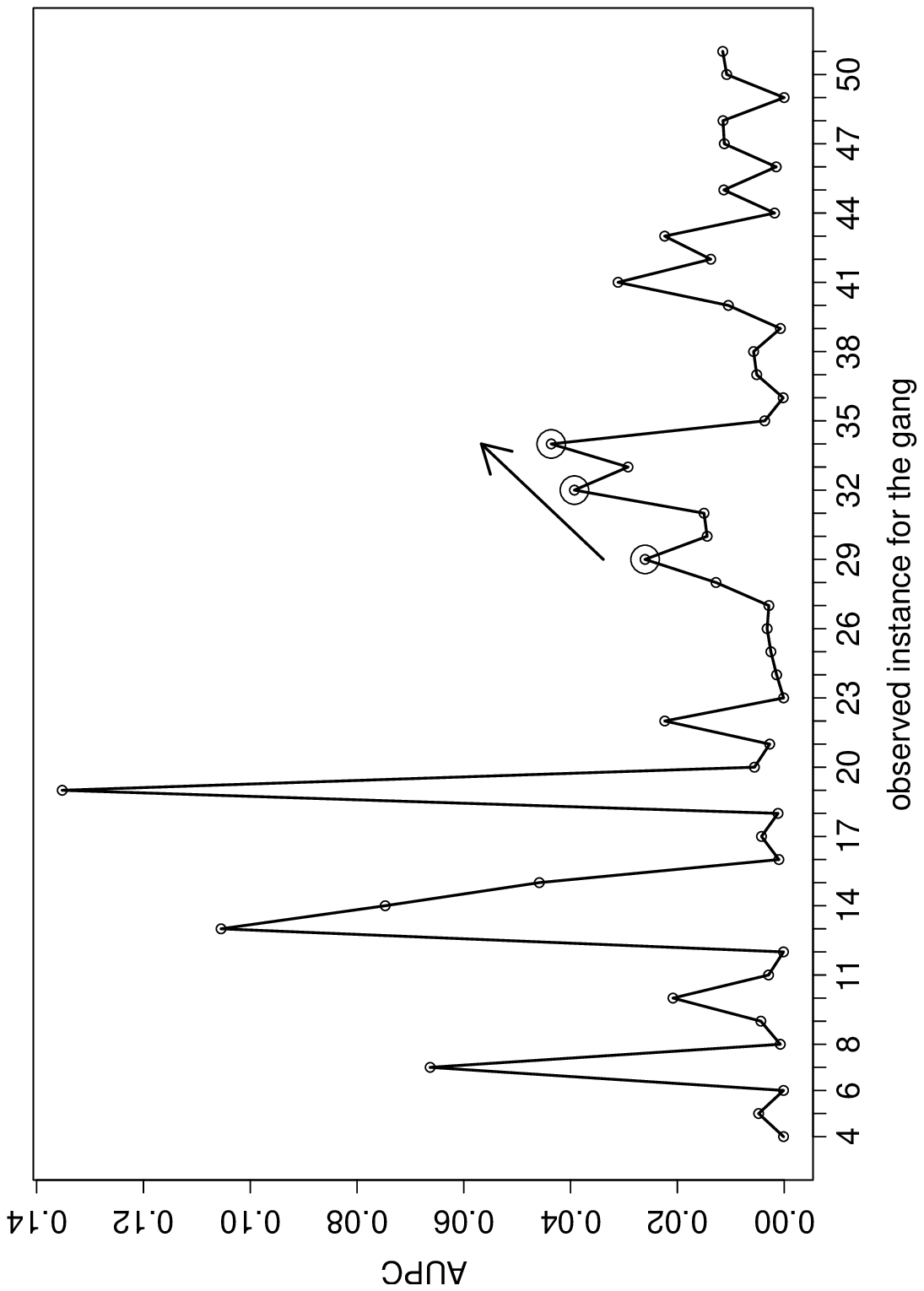}
~\end{subfigure}
\end{figure}

 \begin{figure}[ht]
\center
\caption{\label{F:proximity-curve_A29_34} Proximity curves for the 29th and 34th observed instances of gang A }
\includegraphics[width=7cm, height=7cm, angle=0]{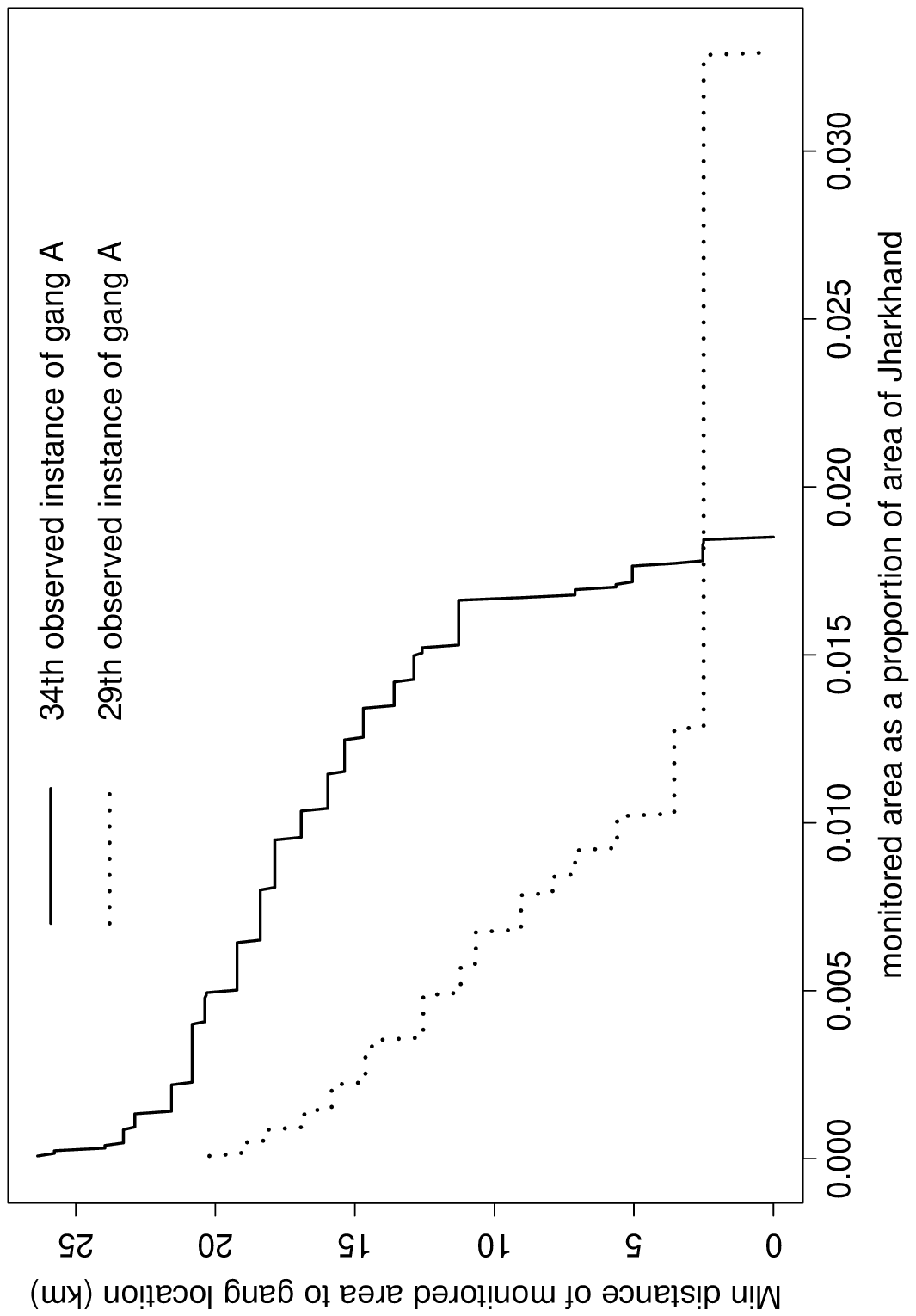}
\end{figure}

\begin{figure}[ht]
\center
\caption{\label{F:with_no_prior_AUPCRAM} Comparison of model-with-prior and model-without-prior, for gang A, with respect to RAM and AUPC.}
\begin{subfigure}{.45\textwidth}
\center
\caption{ highlighting instances where RAM is better for model-with-prior}
\includegraphics[width=6.5cm, height=6.5cm, angle=0]{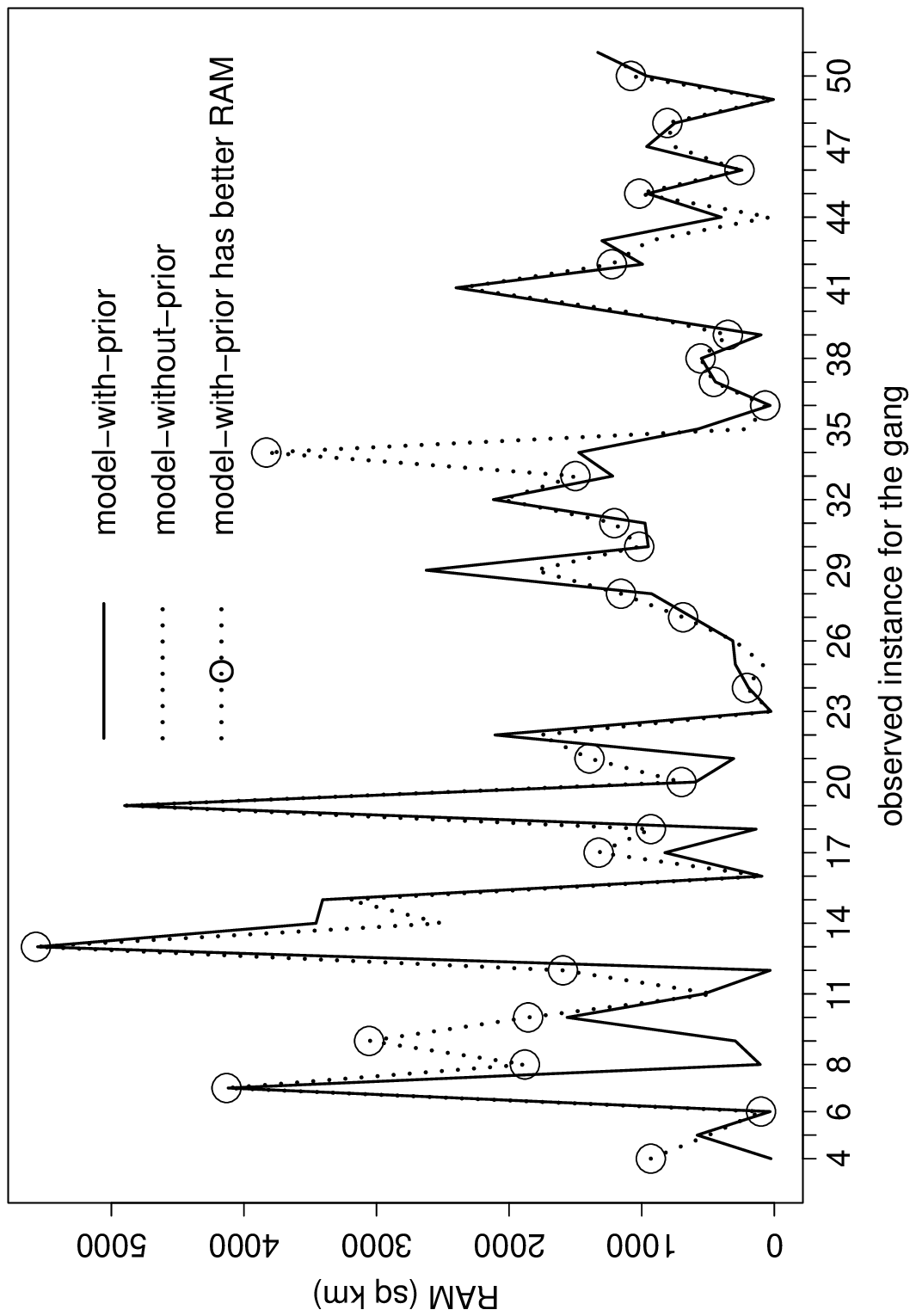}
\end{subfigure}
~\begin{subfigure}{.45\textwidth}
\center
\caption{highlighting instances where AUPC is better for model-with-prior}
\includegraphics[width=6.5cm, height=6.5cm, angle=0]{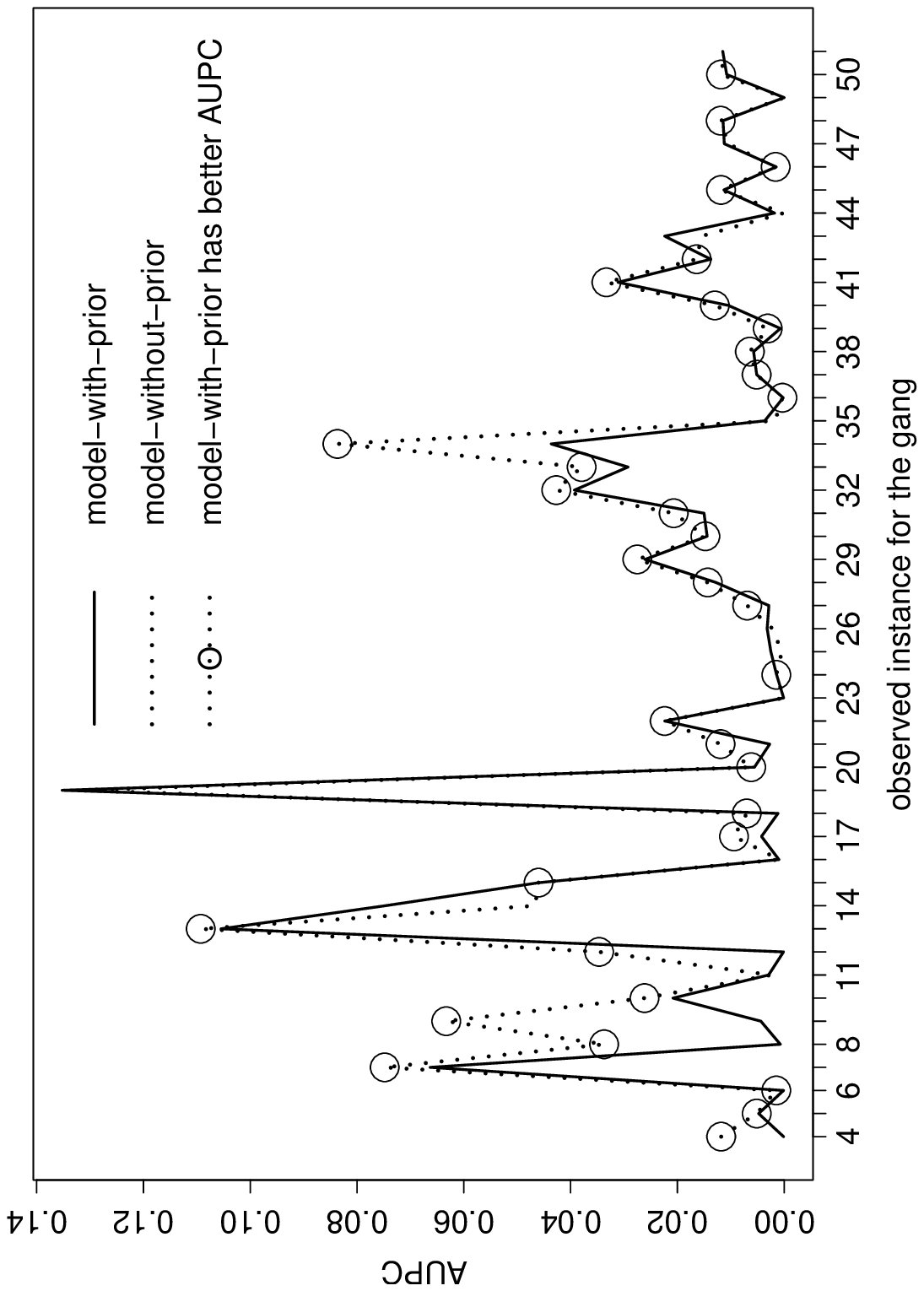}
~\end{subfigure}
\end{figure}

\begin{figure}[ht]
\center
\caption{\label{F:with_no_prior} Prior and posterior weight on expert knowledge}
\includegraphics[width=14cm, height=7cm, angle=0]{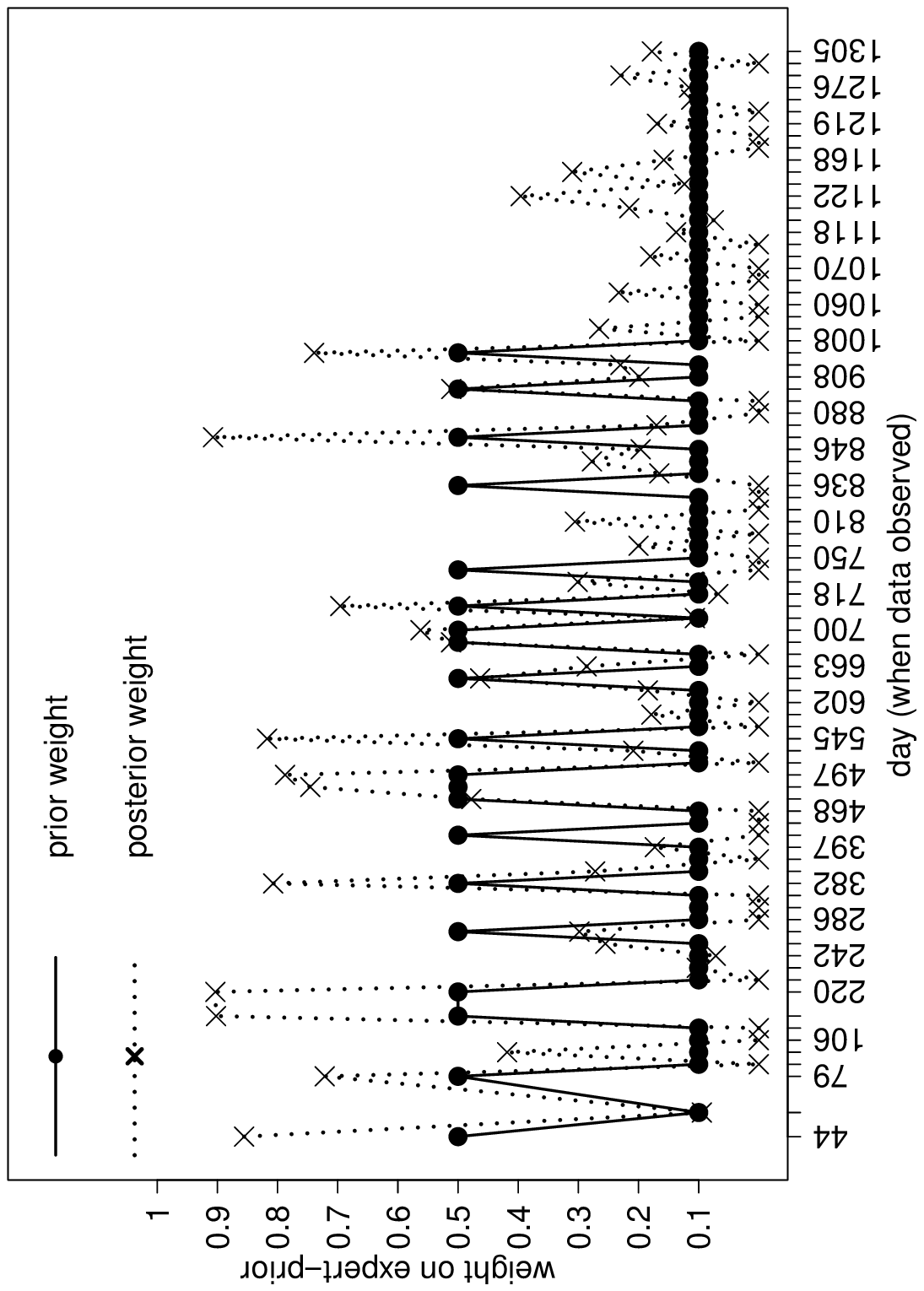}
\end{figure}

\begin{table}[ht]
\centering
\begin{subtable}{.45\textwidth}
\center
\caption{\% instances model-with-prior is better than model-without-prior}
\resizebox{.9\textwidth}{!}{\begin{tabular}{clccc}
  \hline
  gang & instances & RAM & AUPC \\ 
  \hline
  A&  48 & 58\% & 73\% \\ 
  B &  22 & 59\% & 68\% \\ 
  C &  17 & 35\% & 76\% \\ 
   \hline
\end{tabular}}
\end{subtable}
~~\begin{subtable}{.45\textwidth}
\center
\caption{\% instances model-with-prior is at least as good as model-without-prior}
\resizebox{.9\textwidth}{!}{\begin{tabular}{clccc}
  \hline
  gang & instances & RAM & AUPC \\ 
  \hline
  A&  48 & 62\% & 77\% \\ 
  B &  22 & 64\% & 68\% \\ 
  C &  17 & 53\% & 82\% \\ 
   \hline
\end{tabular}}
\end{subtable}
\caption{\label{T:GOF_gang}  Comparing model-with-prior and model-without-prior for 3 gangs}
\end{table}  
\subsection{Implementation in practice}
For implementation in practice, the model can be updated every day, based on data available up to the previous day along with the expert-prior relevant for the current day. The prediction from the model can be used to mark likely areas where the gang may be moving and consequently the nearby police camps can be alerted, to consider the prediction while planning their combing or routine patrolling operations for the day. As an illustration,  Figure \ref{F:pred_mapgangA} (i) shows the likely locations for predicting the locations of the three gangs as per our model for the most recent day in our data, viz. day 1305. The past movement locations of the three gangs can be seen as three separate clusters, gang A at the bottom, gang B top right and gang C top left. The red locations highlighted by our model corresponds to top 500 sq km and yellow to top 500-1000 sq km likely locations of the respective gangs on that day. On day 1305, the locations of gangs A and C were not observed. However,  gang B's location was observed and is shown as a blue dot for reference against our prediction.  We can see that the actual location of gang C, in this case, is captured within the red, i.e. the top 500 sq km area. For reference, we also mark the likely locations based purely on  the expert-prior in Figure \ref{F:pred_mapgangA} (ii). We note that the prior  is coarse and can only take at most 4 distinct values by virtue of its construction. So, we mark all the top locations that take the same prior density value in yellow colour. Consequently, the areas marked in Figure \ref{F:pred_mapgangA} (ii) for gangs A, B, C  respectively, are 962, 1737,  1306 sq km.   The model  which combines the expert-prior with the past data, can be seen to be an enhancement,  which further distinguished between highly likely (red) and relatively less likely locations (yellow). 

\begin{figure}[ht]
\center
\caption{\label{F:pred_mapgangA} Top likely locations predicted for gangs A (bottom cluster), B (top right) and C (top left) for day 1305 in the data. (i) For each gang, top 500 sq km is marked  red,  500-1000 sq km marked yellow, and the actual location for gang C on the day is shown as a blue dot. Locations for other gangs were unknown (ii) all locations with highest value for prior is marked yellow (for gangs A, B, C this marks 962, 1737,  1306 sq km of area respectively as yellow)  }
\begin{subfigure}{.45\textwidth}
\center
\caption{As per Model }
\includegraphics[width=7cm, height=7cm, angle=0]{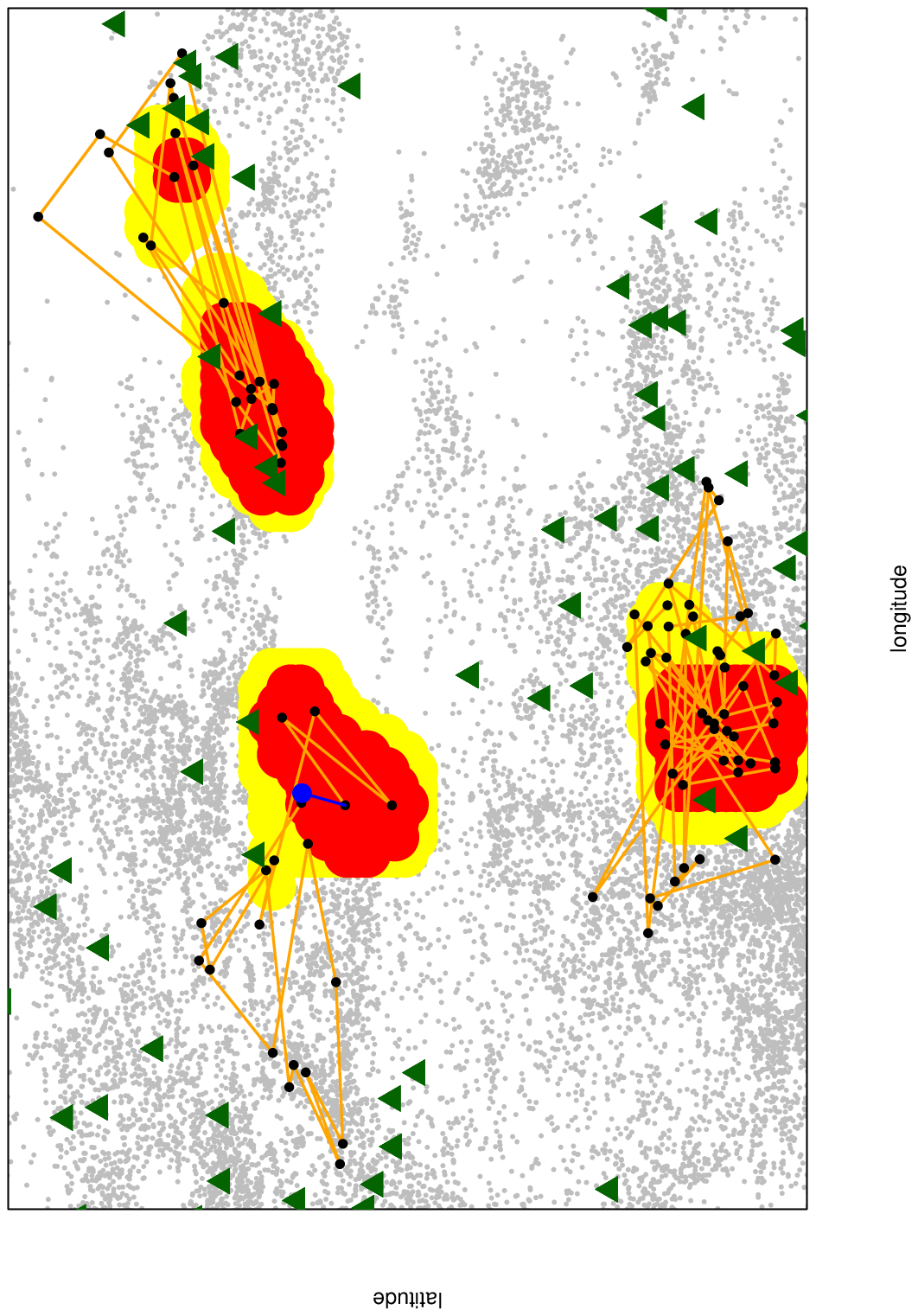}
\end{subfigure}
\begin{subfigure}{.45\textwidth}
\center
\caption{ As per Prior}
\includegraphics[width=7cm, height=7cm, angle=0]{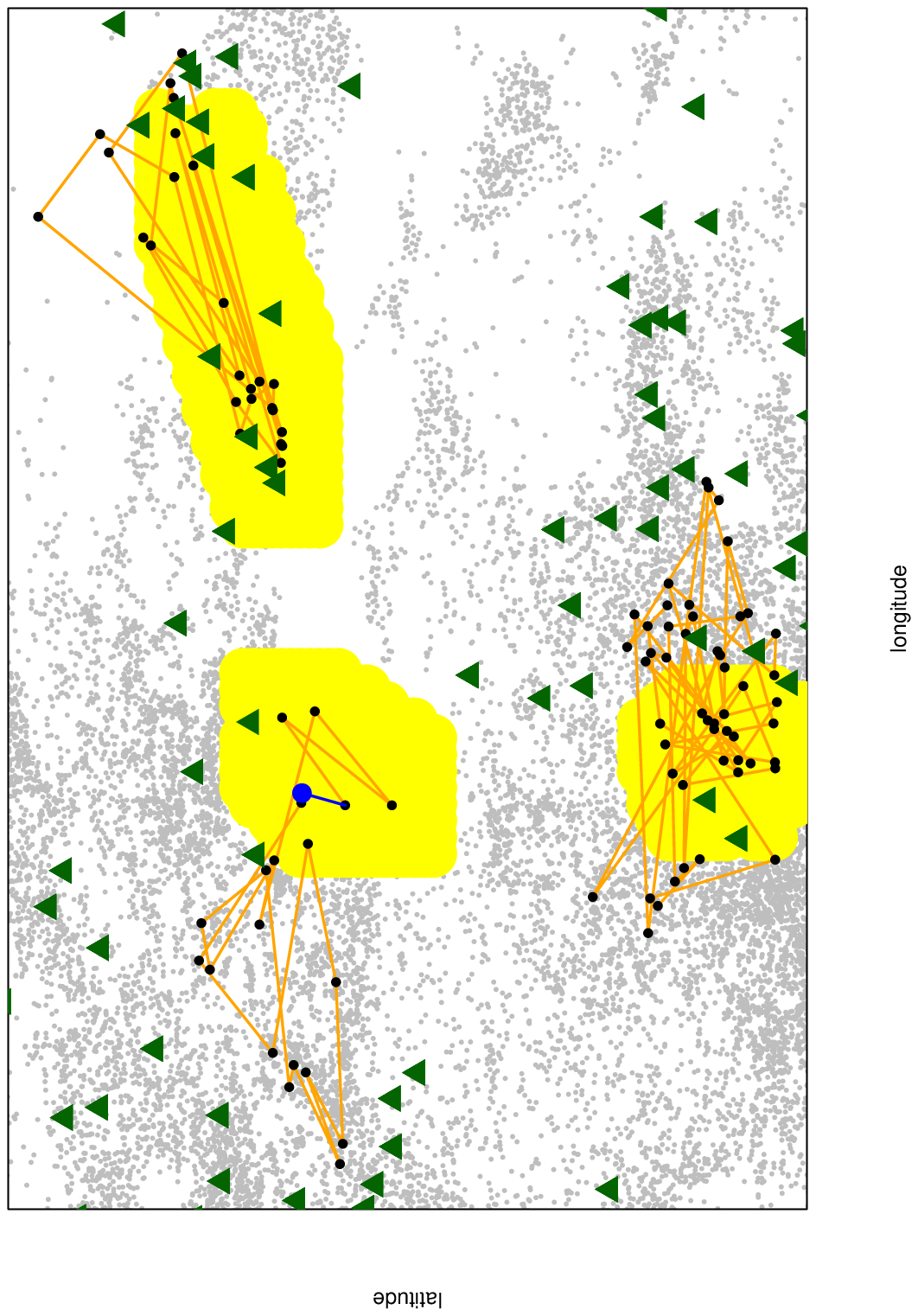}
\end{subfigure}
\end{figure}

\section{Conclusion}
\label{S:conclusion}
We modeled a pressing real-life problem faced by police, which is to predict the movement of insurgent gangs. Specifically, we used real data and expert knowledge on naxalite insurgency in India to develop a coherent methodology to predict movement of insurgent gangs. An important methodological challenge which arises in such problems is that the gang-locations are observed only at irregularly separated time-points.  Formulating and working with a stochastically consistent model specification with such incomplete data is not straight forward. We formulated a weighted kernel density model for which we derived the exact form of likelihood conditional on incomplete past data.  We also provided an approach to construct an expert-prior that accounts for the expert knowledge of police on gang-movements, considering forest density, proximity to CRPF (police) camps, recent observed locations and any additional intelligence available through their informant network. Usually Bayesian analysis requires a prior on the model parameters. Often in practice, as in our context, the prior from expert knowledge is stated in terms of observables (i.e. possible gang-locations) rather than on the model parameters. We provide a systematic approach to incorporate such expert-prior along with a credibility weight, to carry out Bayesian estimation and prediction. Our approach is a first step to constructing such priors, which can be refined with better information gathering mechanisms, post the implementation of the model.

We applied the methods for predicting the movements of three major insurgent gangs that operate in the state of Jharkhand, India. We saw that the model-with-prior usually gave better results than the model-without-prior. We also note that minor changes in the cutoffs of the priors for the forest density or distance to CRPF camps or the credibility weights do not drastically change the results. For predictive assessments, we used one known metric (Required Area  to Monitor) but also suggested a new metric (Area Under Proximity Curve) which provides additional insight. 
 
For practical implementation,  our model can be updated every day based on latest known data and current expert knowledge, and used to highlight likely areas of gang movements on a map for police-use. In many instances, we saw that this can reduce the focus area of gangs' locations from as large as 5500 sq km (say the known area of gang's past operations) to less than 500 sq km.  Although the highlighted areas may appear large, in the insurgency context, this kind of focus is significant from policing operations point of view given that the gangs can move over tens of thousands of square kilometers, and large number of troops being available for anti-insurgency operations.   Combined with the newer aerial techniques for surveillance such as drones, which are now being adopted by the police, the predictions from the model can further help focused patrols and intelligent search.  Aided by the model, the police should develop systems to  observe and collect more frequent data on gang-movements, which in turn can expedite improvements in the model.  

Finally, this methodology may also be potentially  extended to model the activities of other types of organized criminal gangs such as those dealing with narcotics, fake currency or illicit arms because even in these cases, police is able to only tabulate the incidents where either the material is confiscated or some members of the gangs are arrested.

 \FloatBarrier
 \appendix
 \section{Landuse modeling to classify any given location as forest versus non-forest}
 \label{S:landuse_model} 
 Here, we describe the steps involved in developing the model to classify any given location (i.e. longitude, latitude) as forest or non-forest. The exercise has mainly two steps, first to obtain satellite image data from the U.S. Geological Survey \href{https://earthexplorer.usgs.gov}{USGS} website,  and then build a classification model in R (\citealt{cran}) using the `superclass' package. The USGS contains satellite images of a specific dimension recorded every day covering all locations on earth.  Each landsat image is given an identifier by the USGS website, e.g. 141-043, indicating that it is the image of a specific region. There can be several landsat images of the same region taken every day, some may have a lot of cloud cover, some with better clarity. Using our best judgement, we chose the images with good clarity within the same date range as our data. We note that a single landsat image cannot cover the whole state of Jharkhand.  Specifically, as shown in Figure \ref{F:landsat_cover}, we select a polygon area that surrounds the state of Jharkhand on the map, and then select a set of 14 landsat images that together cover the polygon. 
 \begin{figure}[ht]
\center
\caption{\label{F:landsat_cover}  Landsat images that cover the selected polygon area around Jharkhand. Images are labelled 142-043, 142-044, 141-043, 141-044, 141-045, 140-042, 140-043, 140-044, 140-045, 139-042, 139-043, 139-044, 139-045}
\includegraphics[width=10cm, height=5cm]{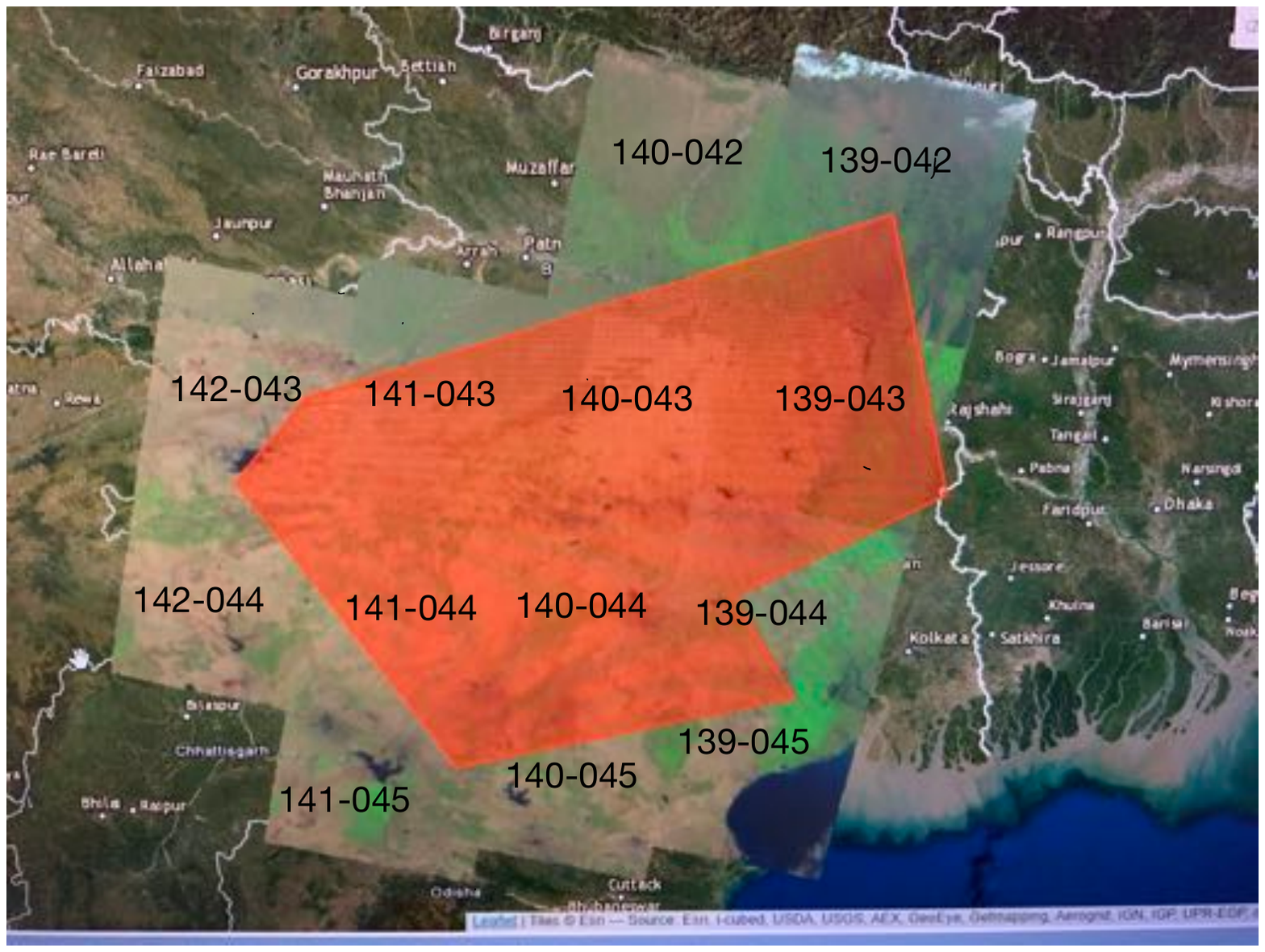}
~\end{figure}

The classification rule for forest versus non-forest that we develop uses the underlying metadata on content of different spectral band levels in the images.  We have used Landsat 8 Operational Land Imager (OLI) and Thermal Infrared Sensor (TIRS) images, which consist of nine spectral bands.  Band 1 represents Coastal/Aerosol, Band 2 indicates Blue, Band 3 represents Green, Band 4 represents Red, Band 5 represents Near InfraRed, Band 6 and Band 7 represents Shortwave Length Infrared, Band 8 represents Panchromatic, Band 9 represents Cirrus, Band 10 and Band 11 represents Thermal Infrared (refer \href{https://www.usgs.gov/faqs/what-are-band-designations-landsat-satellites}{band designations}). This band-level metadata can be combined in R to create interactive images, that can be used to create training data on forest locations and non-forest locations. To help easily identify forest and non-forest, it helps to enhance different spectral combinations in the interactive map. We note that to identify forest areas in the interactive map, we use the infra-red color coding, which can be obtained by using a specific RGB (Reg-Green-Blue) level combination. Since green foliage easily reflects infra red rays, the infra-red color coding can be used to identify thick vegetation and hence forests based whenever the color is dark red.  Similarly, a different RGB combination is used to create an image to enhance non-forest areas to enable the choice of non-forest locations for the purpose of training data creation. Figure \ref{F:choosetr} illustrates the choosing training data points, i.e.pixel locations that  are forest and that are non-forest. 

\begin{figure}
\caption{\label{F:choosetr}Illustration on choosing training data points indicating forest and nonforest on an interactive map. The circles indicate locations chosen as part of the training data.}
\center
\begin{subfigure}{.45\textwidth}
\caption{forest}
\includegraphics[width=10cm, height=5cm]{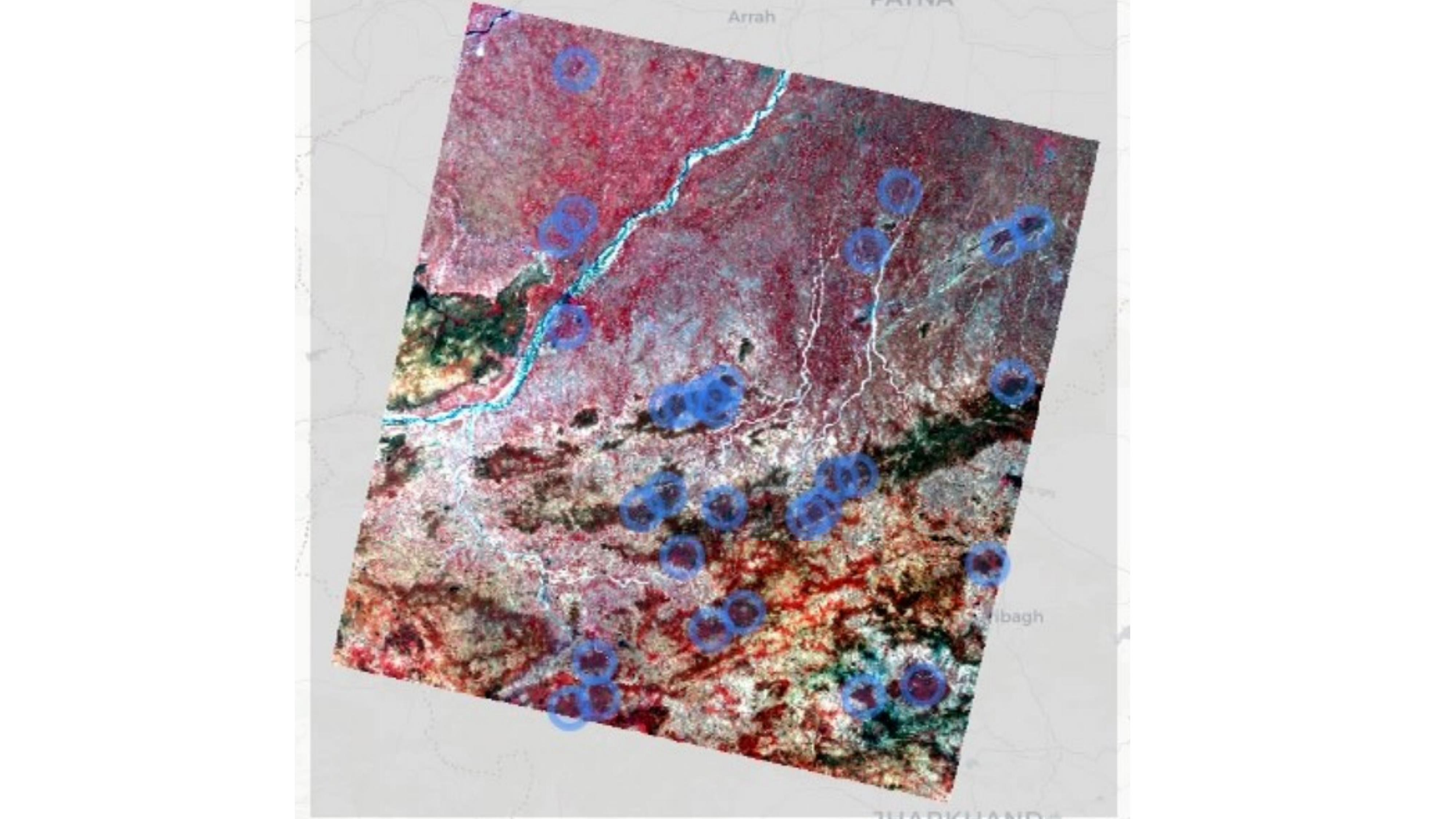}
\end{subfigure}
\begin{subfigure}{.45\textwidth}
\caption{nonforest}
\includegraphics[width=10cm, height=5cm]{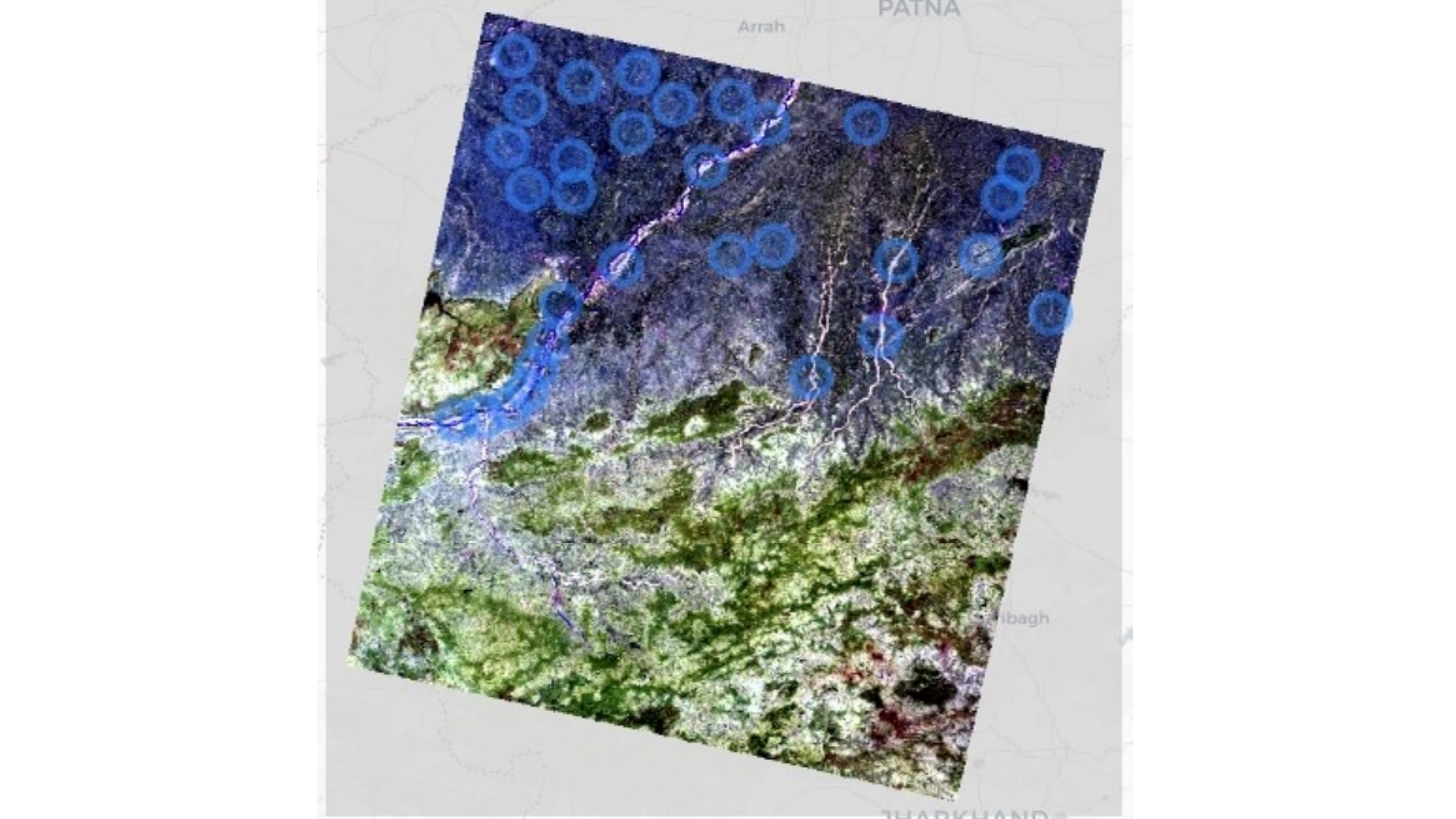}
\end{subfigure}
\end{figure}

 We prepare a training dataset for each of the14 landsat images and develop a classification model to classify different locations within the image as either forest or non-forest, given the  latitude and longitude using the `superclass' package in R (\citealt{cran}). Then, for each point in a fine grid of points covering each landsat image, we apply the classification model to classify the point as forest or non-forest. Figure \ref{F:landuse141043_pred_ggmap} illustrates our  prediction of forest over a fine grid of points  for one of the satellite images, namely 141-043. In  Figure \ref{F:landuse141043_pred_ggmap} part(i), dense green areas are forest areas and in  Figure \ref{F:landuse141043_pred_ggmap} part(ii) black dots indicate locations predicted as forest by the landuse model. The training accuracies of our classification models for the different landsat images covering Jharkhand were between 90 and 95\%. Once a fine grid of points covering the Jharkhand state is assigned a forest or non-forest indicator,  we can compute the forest density at any given location as the percentage of grid points within a 30 km radius that are classified as forest.  The choice of 30 km was fixed based on our understanding from the police given the nature of distances travelled by the criminal gangs. 

 \begin{figure}[ht]
\center
\caption{\label{F:landuse141043_pred_ggmap} Forest prediction on  satellite image 141-043}
\begin{subfigure}{.45\textwidth}
\center
\caption{forests (dark green)}
\includegraphics[width=8cm, height=5cm, angle=0]{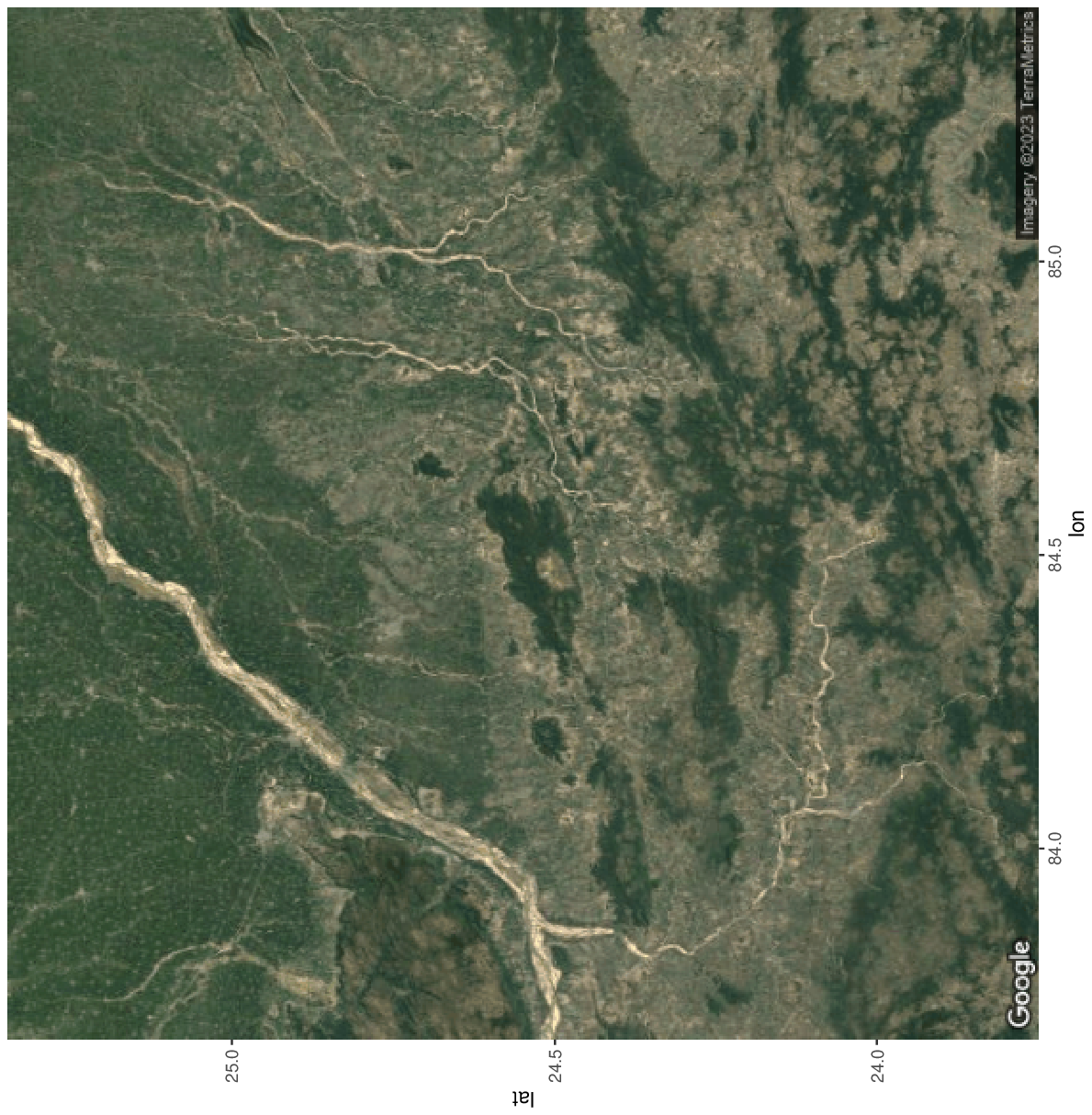}
~\end{subfigure}
\begin{subfigure}{.45\textwidth}
\center
\caption{prediction of forests (black)}
\includegraphics[width=8cm, height=5cm, angle=0]{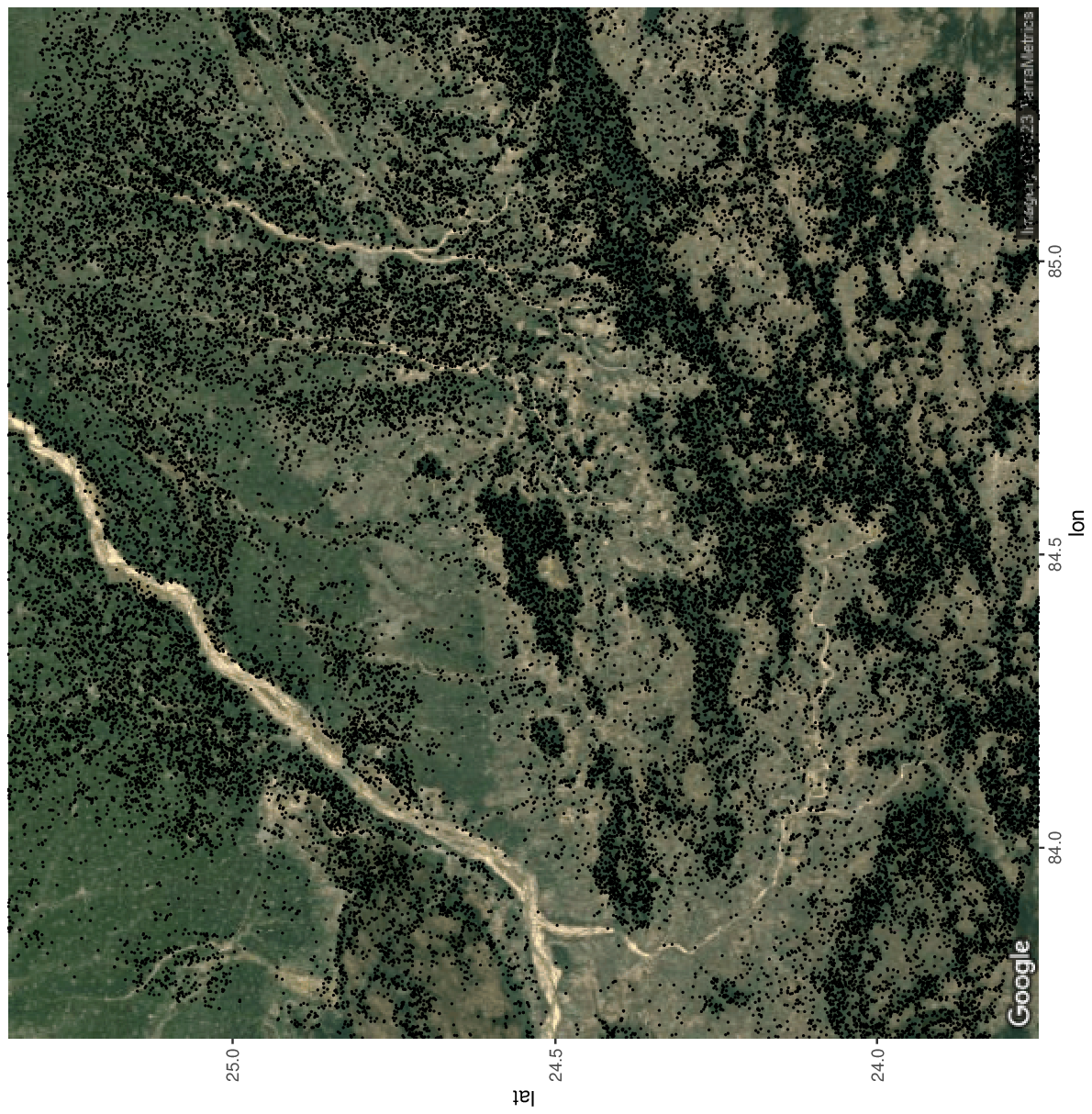}
~\end{subfigure}
\end{figure}

\FloatBarrier
 \section{Proof of Proposition \ref{thm:result1}}
 \label{S:app_proof}
  \begin{proof}[Proof of Proposition \ref{thm:result1}]~\\
We will use mathematical induction on the number of missing past data points ($L$) to prove the result. We first note the following identity. Let $\underline{\tau_1}=\tau_1\cdot (\delta_{lon}, \delta_{lat})$ , $\underline{\tau_2}=\tau_2\cdot (\delta_{lon}, \delta_{lat})$, and $\underline{\tau}=\sqrt{\tau_1^2+\tau^2_2}\cdot (\delta_{lon}, \delta_{lat})$. Then, for $\kappa_2(\cdot, \cdot)$ as in (\ref{eq:bivnormal}),
\begin{eqnarray}
&&\int \kappa_2\left(s-z, \underline{\tau_1}\right)\cdot \kappa_2\left(z-s_i, \underline{\tau_2}\right) dz =  \kappa_2\left(s-s_i, \underline{\tau}\right) .\label{eq:lem1}
\end{eqnarray}
To show equation (\ref{eq:result1}) for $L=1$,  note that using (\ref{eq:lem1}) we have
\begin{eqnarray}
&& f_{(\theta,h)}\left(s_{n+1}=s\vert s_{(1:n)\backslash u_{1}}\right) = \int f_{(\theta,h)}\left(s\vert s_{(1:n)}\right)\cdot f_{(\theta,h)}\left(s_{u_1}\vert s_{(1:u_1-1)}\right)  ds_{u_1}\nonumber \\
 &=& \begin{split} \int \left( \sum_{i\in(1:n)\backslash \{u_1\}} w_\theta(i,n+1) \kappa_2(s-s_i,\underline{h}) + w_\theta(u_1,n+1) \kappa_2(s-s_{u_1},\underline{h}) \right) \\ 
 \cdot   \sum_{j\in(1:u_1-1)} w_\theta(j,u_1) \kappa_2(s_{u_1}-s_i,\underline{h}) ds_{u_1} \end{split}\\
&=&  \sum_{i\in(1:n)\backslash \mathcal{S}_2} w_\theta(i,n+1) \kappa_2(s-s_i,\underline{h}) \nonumber\\
&&+~~~  \sum_{j\in((1:u_1-1)\backslash \mathcal{S}_{1}}  w_\theta(j,u_1)\cdot w_\theta(u_1,n+1) \kappa_2(s-s_{i},\sqrt{2}\underline{h}).
 \end{eqnarray}
 Suppose now that  the following equation holds for some $L<=J$.
 \begin{eqnarray}
 &&~\begin{split}
& f_{(\theta,h)}\left(s_{n+1}=s\vert s_{(1:n)\backslash u_{(1:J)}}\right) = \sum_{i\in(1:n)\backslash\mathcal{S}_{J+1}} w_\theta(i,n+1) \kappa_2(s-s_i,\underline{h}) \\
&+\sum_{m=1}^J \left( \sum_{1\leq i_1 < i_2\cdots<i_m\leq J} w_\theta(u_{i_1},u_{i_2}) \cdot w_\theta(u_{i_2},u_{i_3}) \cdots \right. \\
& \cdots w_\theta(u_{i_{m-1}},u_{i_m}) \cdot w_\theta(u_{i_m},n+1) \cdot \left.  \sum_{j\in(1:u_{i_1}-1)\backslash\mathcal{S}_{i_1}} w_\theta(j,u_{i_1}) \kappa_2(s-s_j,\sqrt{m+1} \cdot \underline{h})\right), 
 \end{split}\nonumber\\
 &&~ \label{eq:result1_J}
 \end{eqnarray}
To obtain the statement for $L=J+1$,  we note that only the first summation term  in (\ref{eq:result1_J}) depends on $s_{u_{J+1}}$, and also that $\mathcal{S}_{J+2}=\mathcal{S}_{J+1} \cup \{u_{J+1}\} $. So, we can write 
 \begin{eqnarray}
 &&f_{(\theta,h)}\left(s_{n+1}=s\vert s_{(1:n)\backslash u_{(1:J+1)}}\right) \nonumber\\
  &=& \int f_{(\theta,h)}\left(s\vert s_{(1:n)\backslash(u_{(1:J)})}\right)  f_{(\theta,h)}\left(s_{u_{J+1}}\vert s_{(1:u_{J+1}-1)\backslash(u_{(1:J)})}\right) ds_{u_{J+1}} \nonumber\\
  &=& ~\sum_{i\in(1:n)\backslash\mathcal{S}_{J+2}} w_\theta(i,n+1) \kappa_2(s-s_i,\underline{h}) \nonumber\\
  && + \int w_\theta(u_{J+1},n+1) \kappa_2(s-s_{J+1},\underline{h})\cdot  f_{(\theta,h)}\left(s_{u_{J+1}}\vert s_{(1:u_{J+1}-1)\backslash(u_{(1:J)})}\right) ds_{u_{J+1}}\nonumber\\
  &&~\begin{split}
&+\sum_{m=1}^J \left( \sum_{1\leq i_1 < i_2\cdots<i_m\leq J} w_\theta(u_{i_1},u_{i_2}) \cdot w_\theta(u_{i_2},u_{i_3}) \cdots w_\theta(u_{i_{m-1}},u_{i_m})\right. \\
& \cdot w_\theta(u_{i_m},n+1) \cdot \left.  \sum_{j\in(1:u_{i_1}-1)\backslash\mathcal{S}_{i_1}} w_\theta(j,u_{i_1}) \kappa_2(s-s_j,\sqrt{m+1} \cdot \underline{h})\right). 
 \end{split}\nonumber\\
 && ~~\label{eq:step2}\hspace{10cm}
  \end{eqnarray}
 Now, note that the second (integration) term in the right hand side of equation (\ref{eq:step2}) can be written as
 \begin{eqnarray}
&& \int w_\theta(u_{J+1},n+1) \kappa_2(s-s_{J+1},\underline{h})\cdot  f_{(\theta,h)}\left(s_{u_{J+1}}\vert s_{(1:(u_{J+1}-1))\backslash(u_{(1:J)})}\right) ds_{u_{J+1}}\nonumber\\
 &=&\int w_\theta(u_{J+1},n+1) \kappa_2(s-s_{u_{J+1}},\underline{h})\cdot \sum_{i\in(1:u_{J+1}-1)\backslash\mathcal{S}_{J+1}} w_\theta(i,u_{J+1}) \kappa_2(s_{u_{J+1}}-s_i,\underline{h}) ds_{u_{J+1}}\nonumber\\
 &&~\begin{split}
&+\int w_\theta(u_{J+1},n+1) \kappa_2(s-s_{u_{J+1}},\underline{h})\cdot \sum_{m=1}^J \left( \sum_{1\leq i_1 < i_2\cdots<i_m\leq J} w_\theta(u_{i_1},u_{i_2}) \cdot w_\theta(u_{i_2},u_{i_3}) \cdots \right. \\
& \cdot w_\theta(u_{i_m},u_{J+1}) \cdot  \left.  \sum_{j\in(1:u_{i_1}-1)\backslash\mathcal{S}_{i_1}} w_\theta(j,u_{i_1}) \kappa_2(s_{u_{J+1}}-s_j,\sqrt{m+1} \cdot \underline{h})\right) ds_{u_{J+1}}
 \end{split}~\nonumber\\
 &=& \sum_{i\in(1:u_{J+1}-1)\backslash\mathcal{S}_{J+1}} w_\theta(i,u_{J+1})\cdot w_\theta(u_{J+1},n+1) \cdot  \kappa_2(s-s_i, \sqrt{2}\cdot \underline{h})\nonumber\\
 &&~\begin{split}
&+\sum_{m=1}^J \left(\sum_{1\leq i_1 < i_2\cdots<i_m\leq J}  w_\theta(u_{i_1},u_{i_2}) \cdot w_\theta(u_{i_2},u_{i_3}) \cdots w_\theta(u_{i_{m-1}},u_{i_m}) \cdot w_\theta(u_{i_m},u_{J+1}) \right. \\
& ~~~~~~~~~~~ \cdot  w_\theta(u_{J+1},n+1) \left.   \sum_{j\in(1:u_{i_1}-1)\backslash\mathcal{S}_{i_1}}   w_\theta(j,u_{i_1})\cdot \kappa_2(s-s_j,\sqrt{m+2} \cdot \underline{h})\right).
 \end{split}~\nonumber\\
 && ~\label{eq:step3}
\end{eqnarray}
So, using (\ref{eq:step3}), equation (\ref{eq:step2}) can be written as
 \begin{eqnarray}
 &&f_{(\theta,h)}\left(s_{n+1}=s\vert s_{(1:n)\backslash(u_{(1:J+1)})}\right) \nonumber\\
  &&= ~\sum_{i\in(1:n)\backslash\mathcal{S}_{J+2}} w_\theta(i,n+1) \kappa_2(s-s_i,\underline{h}) \nonumber\\
  && +  \sum_{i\in(1:u_{J+1}-1)\backslash\mathcal{S}_{J+1}} w_\theta(i,u_{J+1})\cdot w_\theta(u_{J+1},n+1) \cdot  \kappa_2(s-s_i, \sqrt{2}\cdot \underline{h})\nonumber\\
 &&~\begin{split}
&+\sum_{m=1}^J \left(\sum_{1\leq i_1 < i_2\cdots<i_m\leq J}  w_\theta(u_{i_1},u_{i_2}) \cdot w_\theta(u_{i_2},u_{i_3}) \cdots w_\theta(u_{i_{m-1}},u_{i_m}) \right. \\
& \cdot w_\theta(u_{i_m},u_{J+1}) \cdot  w_\theta(u_{J+1},n+1) \left.   \sum_{j\in(1:u_{i_1}-1)\backslash\mathcal{S}_{i_1}}   w_\theta(j,u_{i_1})\cdot \kappa_2(s-s_j,\sqrt{m+2} \cdot \underline{h})\right)
 \end{split}~\nonumber\\
    &&~\begin{split}
&+\sum_{m=1}^J \left( \sum_{1\leq i_1 < i_2\cdots<i_m\leq J} w_\theta(u_{i_1},u_{i_2}) \cdot w_\theta(u_{i_2},u_{i_3}) \cdots \right. \\ 
&\cdots w_\theta(u_{i_{m-1}},u_{i_m}) \cdot w_\theta(u_{i_m},n+1) \cdot \left.  \sum_{j\in(1:u_{i_1}-1)\backslash\mathcal{S}_{i_1}} w_\theta(j,u_{i_1}) \kappa_2(s-s_j,\sqrt{m+1} \cdot \underline{h})\right).
 \end{split}~\nonumber\\
 &&~\label{eq:step4}
  \end{eqnarray}
  Note that the second and third summation term on the right hand side of (\ref{eq:step4}) can be combined and written as 
  \begin{eqnarray}
 &&~\begin{split}
&\sum_{m=1}^{J+1}\left(\sum_{1\leq i_1 < i_2\cdots<i_m\leq J+1,  i_m=J+1}  w_\theta(u_{i_1},u_{i_2}) \cdot w_\theta(u_{i_2},u_{i_3}) \cdots w_\theta(u_{i_{m-1}},u_{i_m}) \right. \\
& \cdot w_\theta(u_{i_m},n+1) \left.   \sum_{j\in(1:u_{i_1}-1)\backslash\mathcal{S}_{i_1}}   w_\theta(j,u_{i_1})\cdot \kappa_2(s-s_j,\sqrt{m+1} \cdot \underline{h})\right).
 \end{split}~\label{eq:step5}
 \end{eqnarray}
and the last summation term  on the right hand side of (\ref{eq:step4}) can be written as 
 \begin{eqnarray}
   &&~\begin{split}
&\sum_{m=1}^J \left( \sum_{1\leq i_1 < i_2\cdots<i_m\leq (J+1), i_m\ne J+1} w_\theta(u_{i_1},u_{i_2}) \cdot w_\theta(u_{i_2},u_{i_3}) \cdots \right. \\
& \cdots w_\theta(u_{i_{m-1}},u_{i_m}) \cdot w_\theta(u_{i_m},n+1) \cdot \left.  \sum_{j\in(1:u_{i_1}-1)\backslash\mathcal{S}_{i_1}} w_\theta(j,u_{i_1}) \kappa_2(s-s_j,\sqrt{m+1} \cdot \underline{h})\right).
 \end{split}~~\label{eq:step6}
    \end{eqnarray}
 Using equations (\ref{eq:step5}) and  (\ref{eq:step6}), we can write equation (\ref{eq:step4}) as
 \begin{eqnarray}
 &&f_{(\theta,h)}\left(s_{n+1}=s\vert s_{(1:n)\backslash(u_{(1:J+1)})}\right) \nonumber\\
 &=& ~\sum_{i\in(1:n)\backslash\mathcal{S}_{J+2}} w_\theta(i,n+1) \kappa_2(s-s_i,\underline{h}) \nonumber\\
  &&~\begin{split}
&+\sum_{m=1}^{J+1}\left(\sum_{1\leq i_1 < i_2\cdots<i_m\leq J+1}  w_\theta(u_{i_1},u_{i_2}) \cdot w_\theta(u_{i_2},u_{i_3}) \cdots w_\theta(u_{i_{m-1}},u_{i_m}) \right. \\
& \cdot w_\theta(u_{i_m},n+1) \left.   \sum_{j\in(1:u_{i_1}-1)\backslash\mathcal{S}_{i_1}}   w_\theta(j,u_{i_1})\cdot \kappa_2(s-s_j,\sqrt{m+1} \cdot \underline{h})\right),
 \end{split}\nonumber\\
 &&~\label{eq:step7}
  \end{eqnarray}
  thus proving the statement (\ref{eq:result1}) for L=J+1.
 \end{proof}

\section{Comparison of proposed model with partial-model}
\label{S:simulation}
In Remark 1 of the paper, we noted that to model incomplete temporal data,  a computationally simpler weighted kernel density as shown in equation  equation (\ref{eq:PR}) could be considered, which is same as using only the first summation term in our likelihood (\ref{eq:result2}) after renormalisation. Let use refer to this as the partial-model as against our proposed model (\ref{eq:result1}) which involves many other terms. Here, we first carry out a simulation to illustrate that the partial could be inadequate to estimate parameters when in the case of incomplete temporal data. Also, in our specific case study we find that the proposed model could be considered better than using the partial model.

For the simulation exercise, data was simulated using the full model (\ref{eq:fullmod}), with $n=500$, $\theta=4$ and $h=1$. So, the full data is $s_1, s_2, \ldots, s_{500}$. Then, $40\%$ of the time-points were randomly chosen to create simulated incomplete temporal data. The parameters were estimated sequentially on the incomplete temporal data using the proposed likelihood based on incomplete temporal data, i.e. equation (\ref{eq:result1}), and also based on partial likelihood as shown in remark 1 equation (\ref{eq:PR})of the paper. Figure \ref{F:figsim} shows the 95\% credible intervals based on each of these models, i.e. the full likelihood and partial likelihood,  obtained for the parameters $h$ and $\theta$. We can see from the figure that the credible intervals from the partial likelihood model struggles to capture the true value of the parameters (shown as dotted line). On the other hand, the credible intervals based on  likelihood (\ref{eq:result1}), which accounts for incomplete temporal data, do capture the true parameter value and narrow down around the true value with increasing values of $n$.
 \begin{figure}[ht]
\center
\caption{\label{F:figsim}  95\% credible intervals (CI) obtained using simulated data based on our proposed model and the (partial) model shown in Remark 1 (\ref{eq:PR}).}
\begin{subfigure}{.45\textwidth}
\caption{$h$}
\includegraphics[width=7cm, height=7cm,angle=0]{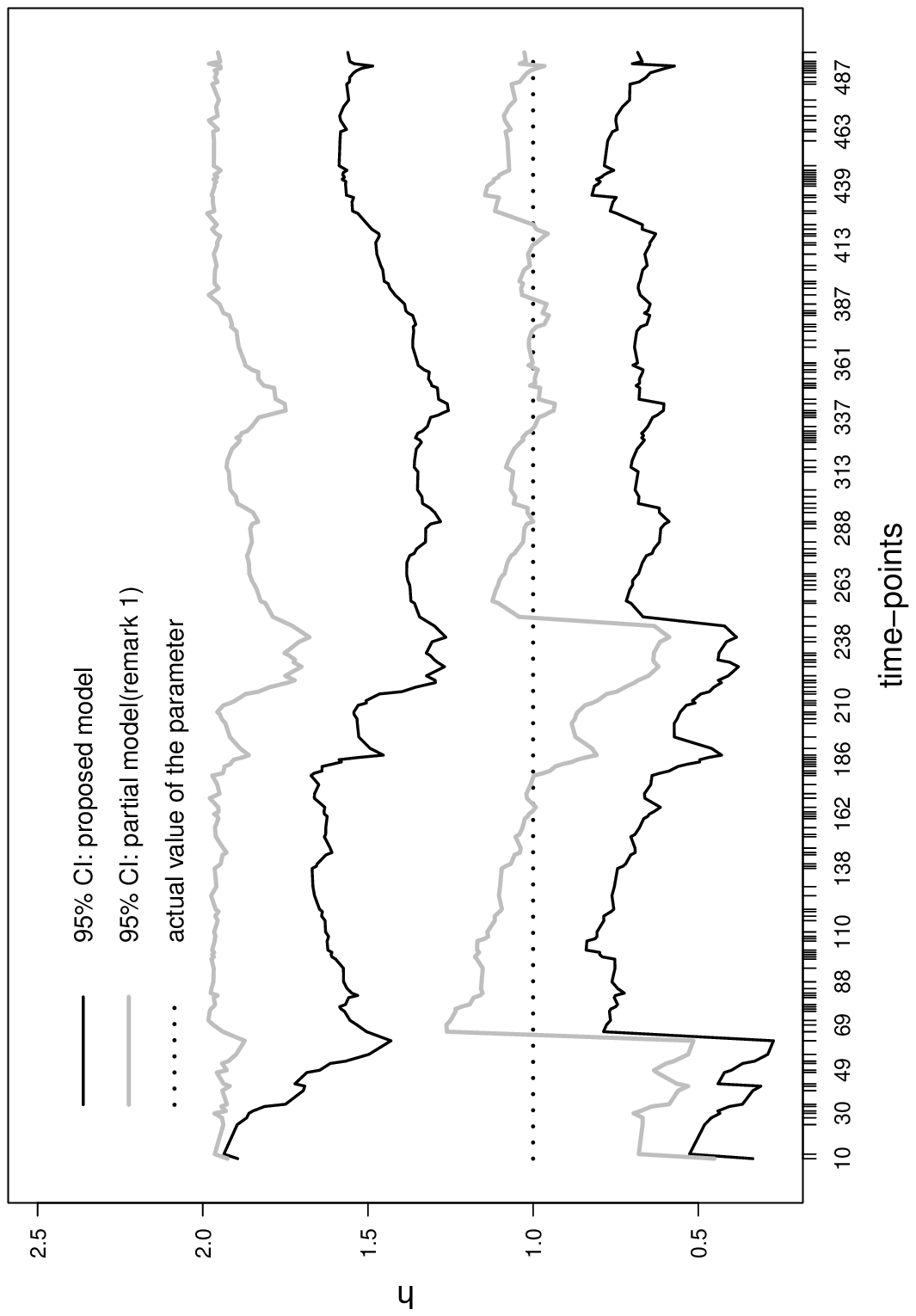}
\end{subfigure}
~\begin{subfigure}{.45\textwidth}
\caption{$\theta$}
\includegraphics[width=7cm, height=7cm,angle=0]{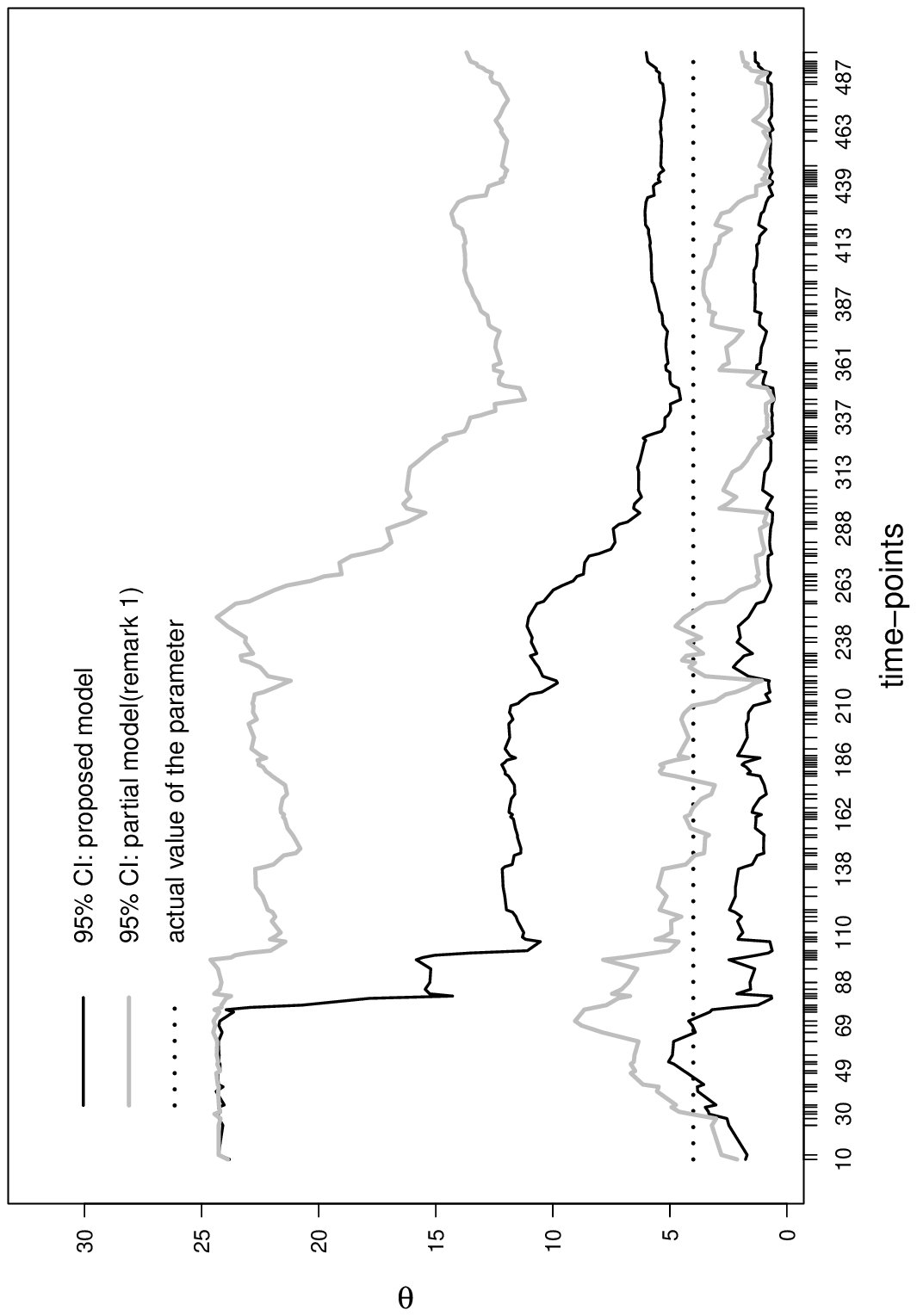}
\end{subfigure}
\end{figure} 

Further, in our data application of predicting movements of naxalite gangs, Figure \ref{F:model_vs_FTonly} compares the performance of our proposed model with the partial-model for gang A. Specifically, for every instance $k$ in the data until the 41st instance, we compute the percentage of instances among $k$th to the 51st instance (i.e. the most recent instance) when our proposed model performed better than the partial-model, (i) in terms of RAM and (ii) in terms of AUPC. There are 51 instances for gang A, but we consider until 41 instances so that we at least have 10 recent instances to compare the two models. While the propsed model does not seem to be performing better initially, if we consider more recent days, the proposed model turns out to be performing better more frequently than the partial-model.
 \begin{figure}[ht]
\center
\caption{\label{F:model_vs_FTonly} Comparison of performance of our proposed model with the partial-model for gang A. For every instance $k$ in the data until the 41st instance, the y-axis shows the percentage of instances among $k$th to the 51st instance (i.e. the most recent instance) when our proposed model performed better than the partial-model, (i) in terms of RAM and (ii) in terms of AUPC  .}
\begin{subfigure}{.45\textwidth}
\caption{RAM}
\includegraphics[width=7cm, height=7cm,angle=0]{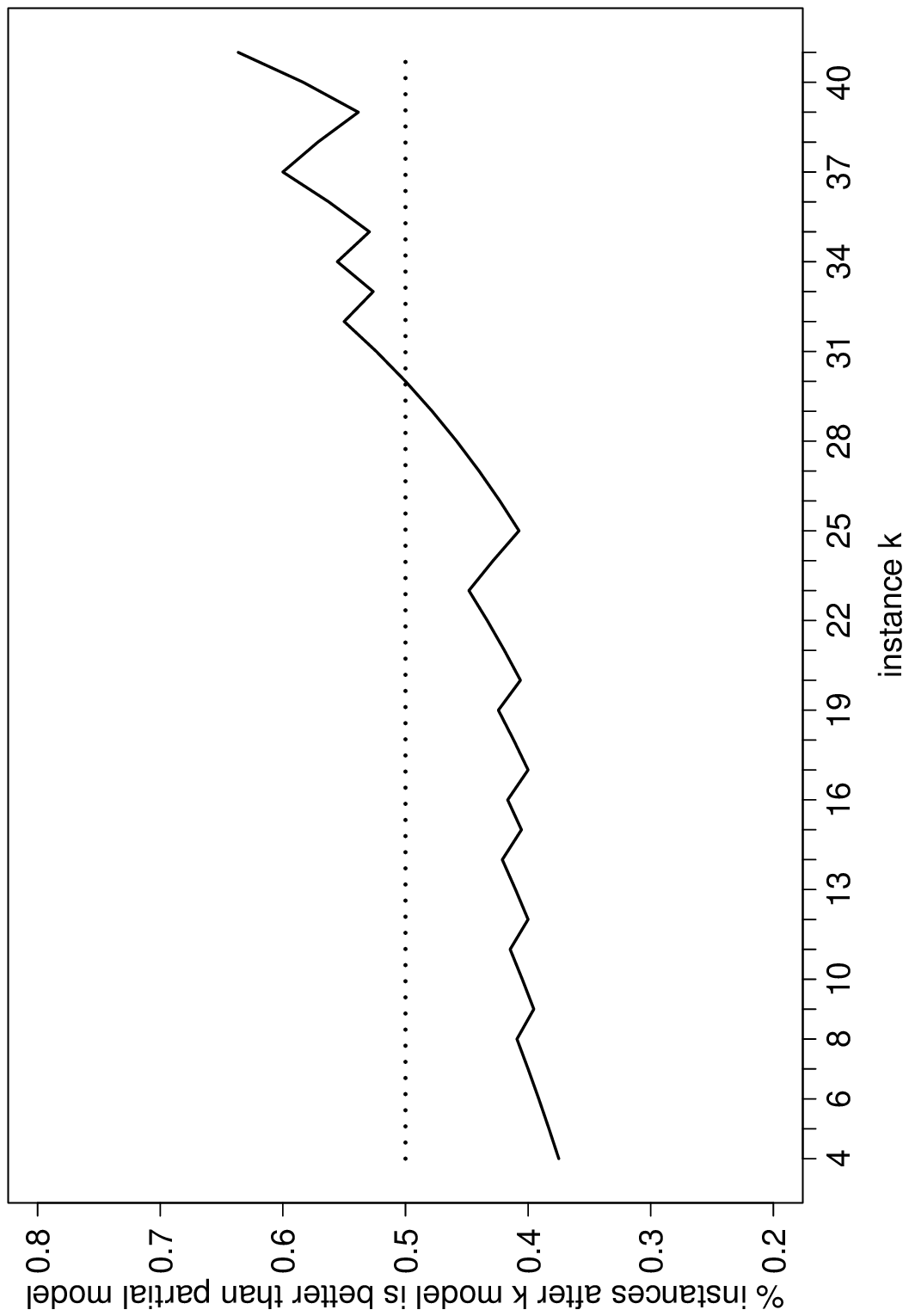}
\end{subfigure}
~\begin{subfigure}{.45\textwidth}
\caption{AUPC}
\includegraphics[width=7cm, height=7cm,angle=0]{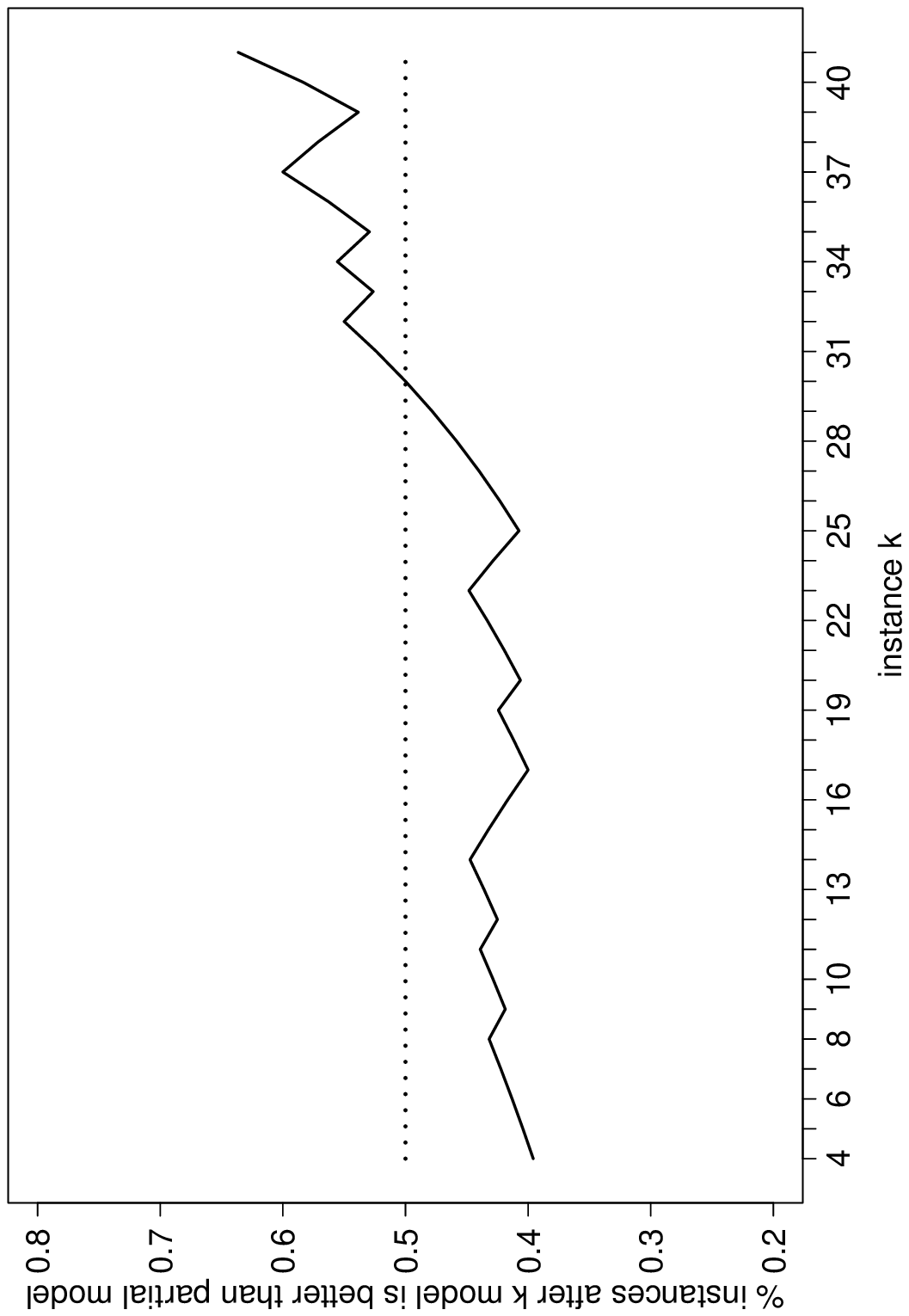}
\end{subfigure}
\end{figure}

\FloatBarrier
\bibliographystyle{rss}
\bibliography{Ref_BDS}
\end{document}